%% file: main.tex
\documentclass[acmsmall,screen,table]{acmart}
\usepackage{natbib}

\usepackage{booktabs}
\usepackage{subcaption}
\usepackage{siunitx}

\usepackage[linesnumbered,ruled,vlined]{algorithm2e}
\usepackage[utf8]{inputenc}
\usepackage{amsmath}
\usepackage{stmaryrd}
\usepackage{soul}
\usepackage{empheq}
\usepackage[breakable]{tcolorbox}
\usepackage{multirow}
\usepackage{wrapfig}
\usepackage{fancyvrb}
\usepackage{listings}
\usepackage{thmtools,thm-restate}
\usepackage{hhline}
\usepackage{setspace}
\usepackage[normalem]{ulem}
\usepackage{multicol}
\usepackage{graphics}
\usepackage{tcolorbox}
\tcbuselibrary{listings, breakable}
\usepackage{array,booktabs}
\usepackage{enumitem}

\SetCommentSty{mycommfont}
\newcommand{\mypar}[1]{\vspace{1mm}\textit{#1.}}

\newcommand{\bigex}{%
\mathop{\lower0.75ex\hbox{%
   \scalebox{1.7}{\ensuremath{\exists}}}}\limits}

\xspaceaddexceptions{[]\{\}}
\input{macro.tex}

\usepackage{tikz}
\usetikzlibrary{shapes,snakes}
\usetikzlibrary{arrows.meta}
\usetikzlibrary{positioning}
\usetikzlibrary{shapes.geometric}

\definecolor{darkcolor}{HTML}{efefef}
\definecolor{litecolor}{HTML}{c4eded} 

\usepackage[framemethod=tikz]{mdframed}
\usepackage{pgfplots}

\usepackage[nosort]{cleveref}

\usepackage{csquotes}
\usepackage{framed}
\usepackage{arydshln}
\usepackage{afterpage}
\usepackage{enumitem}

\newtheorem{definition}{Definition}[section]
\setlist[itemize]{align=parleft,left=0pt..1em, topsep=2pt}
\setlist[description]{topsep=2pt}

\setcopyright{cc}
\setcctype{by}
\acmDOI{10.1145/3776722}
\acmYear{2026}
\acmJournal{PACMPL}
\acmVolume{10}
\acmNumber{POPL}
\acmArticle{80}
\acmMonth{1}
\received{2025-07-10}
\received[accepted]{2025-11-06}

\begin{document}

\title{Nice to Meet You: Synthesizing Practical \changed{MLIR} Abstract Transformers}

\author{Xuanyu Peng}
\authornotemark[1]
\orcid{0000-0001-8613-3506}
\affiliation{%
  \institution{University of California San Diego}
  \country{USA}
}
\email{xup002@ucsd.edu}

\author{Dominic Kennedy}
\authornote{Equal contribution}
\orcid{0000-0001-7368-4333}
\affiliation{%
  \institution{University of Utah}
  \country{USA}
}
\email{dominicmkennedy@gmail.com}

\author{Yuyou Fan}
\orcid{0009-0005-5742-0692}
\affiliation{%
  \institution{University of Utah}
  \country{USA}
}
\email{yuyou.fan@utah.edu}

\author{Ben Greenman}
\orcid{0000-0001-7078-9287}
\affiliation{%
  \institution{University of Utah}
  \country{USA}
}
\email{benjamin.l.greenman@gmail.com}

\author{John Regehr}
\orcid{0000-0001-7025-4610}
\affiliation{%
  \institution{University of Utah}
  \country{USA}
}
\email{regehr@cs.utah.edu}

\author{Loris D'Antoni}
\orcid{0000-0001-9625-4037}
\affiliation{%
  \institution{University of California San Diego}
  \country{USA}
}
\email{ldantoni@ucsd.edu}

\input{abstract.tex}

\ccsdesc[500]{Software and its engineering~Compilers}
\ccsdesc[500]{Theory of computation~Abstraction}
\ccsdesc[500]{Theory of computation~Program analysis}
\ccsdesc[500]{Software and its engineering~Automated static analysis}
\ccsdesc[500]{Software and its engineering~Automatic programming}

\keywords{\xuanyuchanged{Abstract Interpretation}, Program Synthesis, LLVM}

\maketitle

\input{1introduction.tex}
\input{2illustration_example.tex}


\input{3framework.tex}

\input{4algorithm.tex}
\input{5implementation}

\input{7evaluation.tex}

\input{8related_work}

\section{Conclusion}
\label{s:conclusion}

Abstract transformers are a load-bearing component of a modern optimizing compiler: the compiler will miss optimizations if transformers are imprecise, and it will miscompile if they are unsound.
We created \name: a framework for synthesizing formally verified abstract transformers from specifications of integer IR instructions and \changed{finite} non-relational abstract domains.
Unlike previous systems, ours does not require any sketches---transformers are synthesized from scratch---and can therefore quickly synthesize transformers for dozens of operators. 
The insight that made this possible is that transformers can be synthesized piecewise, with each new piece targeting a different part of the input space.
The final transformer is simply the meet of its constituent pieces.
In our evaluation,
\name synthesized transformers for most LLVM operations with precision sometimes comparable to \changed{LLVM's manually written transformers}.
\name also synthesized \benchanged{26} transformers that are either more precise than LLVM's or can be combined with LLVM ones via a meet operation to yield new transformers with greatly increased precision---i.e., \name's transformers deal with corner cases that had eluded LLVM developers.

\section*{Data-Availability Statement}
\name{} and scripts for reproducing our experiments
are available on Zenodo~\cite{xsynth-zenodo}.

\begin{acks}
  Kennedy was supported by a seed grant from the University of Utah.
D'Antoni and Peng are supported in part by a Microsoft Faculty Fellowship; a UCSD JSOE Scholarship; and NSF under grants CCF-2422214, CCF-2506134 and CCF-2446711. 
Fan and Regehr were supported in part by the National Science Foundation under grant under Grant No.\ 1955688.
Any opinions, findings, and conclusions or recommendations expressed in this publication are those of the authors, and do not necessarily reflect the views of the sponsoring entities.
Loris D'Antoni holds concurrent appointments as a Professor at the University of California San Diego and as an Amazon Scholar. This paper describes work performed at the University of California San Diego and is not associated with Amazon.
\end{acks}

\bibliographystyle{ACM-Reference-Format}
\bibliography{main.bib}

\whenappendix{
\appendix
\input{10appendix-amurth}
\input{11appendix-ablation}
}

\end{document}

%% file: macro.tex
\newcommand{\rone}{(\emph{i})\xspace}
\newcommand{\rtwo}{(\emph{ii})\xspace}
\newcommand{\rthree}{(\emph{iii})\xspace}

\let\oldnl\nl
\newcommand{\nonl}{\renewcommand{\nl}{\let\nl\oldnl}}

\newcommand{\benchanged}[1]{#1} 
\newcommand{\changed}[1]{#1} 
\newcommand{\xuanyuchanged}[1]{#1} 
\newcommand{\Omit}[1]{}

\newcommand{\cdom}{\mathcal{C}}
\newcommand{\adom}{\mathcal{A}}

\newcommand{\cdomp}[1]{\cdom_{#1}}
\newcommand{\adomp}[1]{\adom_{#1}}
\newcommand{\tf}{f^{\#}}

\newcommand{\besttf}{\widehat{f}^{\#}}
\newcommand{\settf}{\mathcal{F}}
\newcommand{\settfp}[1]{\mathcal{#1}}
\newcommand{\meettf}{\settf_{\sqcap}}
\newcommand{\meettfp}[1]{\settfp{#1}_{\sqcap}}

\newcommand{\powerset}[1]{\mathcal{P}(#1)}
\newcommand{\condtf}{\tf_c}

\newcommand{\soundset}{\settf^s}
\newcommand{\precset}{\settf^p}
\newcommand{\soundsetmeet}{\settf^s_{\sqcap}}

\newcommand{\preccand}{C^p}
\newcommand{\lang}{\mathcal{L}}
\newcommand{\langop}{\mathcal{L}_{op}}

\newcommand{\nstep}{N_{\text{step}}}

\newcommand{\mcmctemp}{T}

\newcommand{\fcpair}[2]{\texttt{ite(}#1\texttt{, }#2\texttt{, }\top\texttt{)}}

\newcommand{\costfunc}{\textnormal{{\texttt{Cost}}}\xspace}
\newcommand{\sound}{\textnormal{{\texttt{Soundness}}}\xspace}

\newcommand{\improvement}{\textnormal{{\texttt{Improvement}}}\xspace}
\newcommand{\norm}[1]{\left\lVert {#1} \right\rVert}
\newcommand{\size}[1]{\left\lvert {#1} \right\rvert}
\newcommand{\normp}[2]{\left\lVert {#1} \right\rVert_{#2}}
\newcommand{\amurth}{\textnormal{{\textsc{Amurth}}}\xspace}
\newcommand{\knownbits}{\texttt{KnownBits}\xspace}
\newcommand{\knownzero}{\texttt{Zero}\xspace}
\newcommand{\knownone}{\texttt{One}\xspace}
\newcommand{\constantrange}{\texttt{ConstantRange}\xspace}
\newcommand{\issound}{\texttt{sound}}
\newcommand{\soundat}{\texttt{soundAt}}
\newcommand{\ucr}{\texttt{CR}_u}
\newcommand{\scr}{\texttt{CR}_s}
\newcommand{\findnewtf}{\texttt{SynthesizeTransformer}\xspace}

\newcommand{\apint}{\texttt{APInt}\xspace}

\newcommand{\name}{\texttt{NiceToMeetYou}\xspace}
\newcommand{\problem}{prob}

\newcommand{\basicops}{\textbf{Basic} }
\newcommand{\bitops}{\textbf{BitExt} }

\definecolor{dred}{rgb}{0.6,0.0,0.0}
\definecolor{dblue}{rgb}{0.0,0.0,0.5}
\definecolor{dgreen}{rgb}{0.0,0.5,0.0}
\definecolor{gray}{rgb}{0.5,0.5,0.5}

\lstdefinelanguage{MLIR}{
    keywords={func, return},
    keywordstyle=\color{dblue}\bfseries,
    morekeywords={[2]KnownBits, bool},
    keywordstyle={[2]\color{dgreen}\bfseries},
    morekeywords={[3]negate, countLeadingZero, setHighBits, makeKnownBits, unsignedLessThan, ite, meet},
    keywordstyle={[3]\color{dred}},
    sensitive=true,
    comment=[l]{//},
    commentstyle=\color{gray},
    basicstyle=\ttfamily\footnotesize,
    escapeinside=``,
    tabsize=3,
    numbers=none,
    xleftmargin=0em
}

\lstdefinelanguage{llvm}{
    morekeywords={[1]if, return},
    keywordstyle={[1]\color{dblue}\bfseries},
    morekeywords={[2]KnownBits, APInt, unsigned, bool},
    keywordstyle={[2]\color{dgreen}\bfseries},
    morekeywords={[3]countLeadingZero, countTrailingZero, setHighBits, setLowBits, isConstant, getConstant, isPowerOf2, max},
    keywordstyle={[3]\color{dred}},
    sensitive=true,
    comment=[l]{//},
    commentstyle=\color{gray},
    basicstyle=\ttfamily\footnotesize,
    escapeinside=``,
    tabsize=3,
    numbers=none,
    xleftmargin=0em
}

\newcommand{\cpp}{\texttt{C++}}
\newcommand{\pct}[1]{\(#1\,\%\)}
\newcommand{\abs}[1]{\lvert #1 \rvert}

\newcommand{\eqdef}{\mathrel{\overset{\mathrm{def}}{=}}}

\let\showappendix\relax 
\newcommand{\ifappendix}[2]{\ifdef{\showappendix}{#1}{#2}}
\newcommand{\whenappendix}[1]{\ifappendix{#1}{}}

%% file: abstract.tex
\begin{abstract}
Static analyses play a fundamental role during compilation: they
discover facts that are true in all executions of the code being
compiled, and then these facts are used to justify optimizations and
diagnostics.
Each static analysis is based on a collection of \emph{abstract
transformers} that provide abstract semantics for the concrete
instructions that make up a program.
It can be challenging to implement abstract transformers that are
sound, precise, and efficient---and in fact both LLVM and GCC have
suffered from miscompilations caused by unsound abstract transformers.
Moreover, even after more than 20 years of development, LLVM lacks
abstract transformers for \benchanged{hundreds of} instructions in its intermediate
representation (IR).

We developed \name: a program synthesis framework for abstract
transformers that are aimed at the kinds of non-relational integer
abstract domains that are heavily used by today's production
compilers.
It exploits a simple but novel technique for breaking the synthesis
problem into parts: each of our transformers is the meet of a
collection of simpler\changed{, sound} transformers that are synthesized
such that each new piece fills a gap in the precision of the final transformer.
Our design point is bulk automation: no sketches are required.
\benchanged{Transformers are verified by lowering to a previously-created
SMT dialect of MLIR\@.}
Each of our synthesized transformers is provably sound
\benchanged{and some~($17\,\%$) are more precise than those provided by LLVM\@.}
\end{abstract}

%% file: 1introduction.tex
\section{Introduction}
\label{Se:Introduction}

A modern, highly optimizing compiler runs numerous dataflow analyses on the code that is being compiled; the results of the analyses are used to justify optimizations and diagnostics. 
For example, LLVM relies heavily on a ``\knownbits'' analysis that attempts to prove that individual bits of SSA values are either zero or one in every execution of the program being compiled.

Empirically, the process of engineering a dataflow-driven compiler works as follows.
First, engineers recognize the need for dataflow results and implement the basic analysis structure within the compiler, which is initially highly imprecise because not enough, and not-precise-enough, \textit{abstract transformers} have been written.
Then, optimizations and diagnostics driven by analysis results are added, typically alongside improvements to analysis precision that are necessary to make the compiler operate robustly. 
In a compiler like LLVM, where the IR (intermediate representation) instruction set is large
\changed{(over 400 target-independent instructions and intrinsics)},
this process takes an enormous amount of time and energy.
%
Across 17 different LLVM backends, only four have any abstract transformers at all for LLVM instructions representing target-specific intrinsics, and even those four have poor coverage: 30 out of 1713 intrinsics for x86-64, 2 out of 726 for RISC-V, 2 out of 1286 for AMD GPUs, and 5 out of 1673 for AArch64.
Operations that lack abstract transformers must be analyzed conservatively: they return \emph{top}, the unknown value.
This low coverage can lead to unpredictable compilation effects where, for example, developers who substitute intrinsics for portable code while chasing performance can see degraded dataflow-driven optimizations in nearby code because the portable code could be analyzed but the intrinsics cannot.
Moreover, bugs in unverified dataflow analyses have led to miscompilation errors in both GCC and LLVM\@ ~\cite{Regehr, Lewycky}.

Our work attacks the problem of synthesizing abstract transformers from concrete instruction semantics that are formally specified using an MLIR dialect~\cite{mlir-pldi25}.
Our goal is to develop technologies that can rapidly provide compiler developers with reasonable initial implementations.
We validate our prototype by synthesizing transformers for three non-relational, compiler-friendly abstract domains (\knownbits, signed and unsigned \constantrange) for 39 instructions that are present both in LLVM and in MLIR's Arith dialect.
%
These are formally verified to be sound, and in some cases are more precise than those that are part of LLVM's implementation, which has been tweaked for precision by numerous compiler developers over the last 20 years.

Given a concrete operation $f$ and an abstract domain (e.g., \knownbits), our goal is to synthesize a corresponding abstract transformer $\tf$.
To be sound on a given abstract input, $\tf$ must return an abstract value that over-approximates the set of all possible outputs produced by applying $f$ to any concrete inputs described by its abstract inputs. 
The smaller this over-approximation, the more precise the abstract transformer.
Formally,  
when the power set of concrete values $\powerset{\cdom}$ is related to the set of abstract values (i.e., abstract domain) $\adom$ by a Galois connection $\mathcal{P}(\cdom)
    \mathrel{\substack{\xleftarrow{\gamma} \\[-0.6ex] \xrightarrow[\alpha]{}}}
    \adom$, there is a \emph{best abstract transformer} $\besttf$ that is defined as $\besttf = \alpha \circ \tilde{f} \circ \gamma$, where $\tilde{f}$ runs $f$ on a set of concrete values and produces their corresponding concrete outputs~\cite{DBLP:conf/popl/CousotC77}.
    However, this transformer definition does not directly lead to a usable implementation: it requires taking the meet of a set of abstract values whose size is exponential in the bitwidth of the concrete values being analyzed.

Since it does not seem generally practical to synthesize best abstract transformers, previous research efforts have focused on finding efficient approximations of them.
Some approaches, such as \citet{Scherpelz2007Rewrite} and \citet{DBLP:journals/toplas/ElderLSAR14}, have targeted specific abstract domains. 
\citet{DBLP:journals/pacmpl/KalitaMDRR22} provide a synthesis framework that is applicable to arbitrary domains, but it depends on user-provided program sketches~\benchanged{(see~\Cref{se:evaluation:amurth}\whenappendix{~and~\Cref{app:amurth}})}.

Our work addresses the following research questions \changed{in the context of finite, Galois-connection-based abstract domains~(c.f.~\cite{cc-plilp-1992,DBLP:conf/popl/CousotC77})}:
\begin{itemize}
    \item \textbf{Practicality and Generality:} 
    Using existing formally specified concrete instruction semantics, can we automatically synthesize abstract transformers for multiple domains that are used in real-world compilers? In particular, can we generate functions that compiler developers can adopt---i.e., ones that are free of external dependencies, performant enough for production use, and sensitive to IR-level subtleties such as undefined behavior?
    
    \item \textbf{Soundness and Precision:} 
    How precisely can our synthesized transformers approximate the ideal transformer, $\besttf$, while being provably sound?
    
    \item \textbf{Automation:} \
    Can our synthesis procedure navigate the search space without meaningful help from users? In other words, can we succeed without requiring sketches?
\end{itemize}

We treat synthesis as an optimization problem, where the objective is to find a sound transformer $\tf$ that minimizes the user-given norm $\norm{\tf}$ that measures the imprecision for every possible pair of abstract inputs (note that such a function is easy to implement for given abstract domains).
Since the functions that we wish to synthesize are (empirically) out of direct reach for enumeration or CEGIS, and since we do not wish to rely on user-provided sketches, we use stochastic search techniques inspired by Stoke~(\citet{Alex2013Stoke}), where candidate transformers evolve through a sequence of random modifications guided by the cost function induced by our objective.

We have observed two implementation patterns in code produced by compiler developers writing highly precise abstract transformers for GCC and LLVM\@.
Both of these ended up leading to key aspects of our approach:

\paragraph{\changed{Pattern 1:} Splitting the input space}
Practical abstract transformers often gain precision by making a case split to separately handle different parts of the input space.
For example, LLVM's transformer for bitvector truncation on integer ranges\footnote{\url{https://github.com/llvm/llvm-project/blob/release/20.x/llvm/lib/IR/ConstantRange.cpp\#L864-L915}} begins with special cases for $\top$ and $\bot$ and then subsequently splits on whether the incoming integer range wraps around the $\texttt{UINT\_MAX}\ /\ 0$ boundary.
We observed that the logic at the finest granularity was often reasonably simple, but that the overall transfer functions that we wanted to synthesize appeared to be significantly more complicated than anything we could reliably generate without a sketch.

The insight that allowed us to make progress here was that the meet of a collection of sound abstract transformers is still sound.
Thus, we can synthesize an abstract transformer in parts, and then assemble the parts at the end: $\meettf = \tf_1 \sqcap \cdots \sqcap \tf_n$.
If we simply synthesized a pile of transformers and glued them together, they would be likely to all cover similar parts of the input space (because, e.g., some parts of the input space are easier to cover than others).
We discourage this behavior by \emph{dynamically adapting} our fitness function: each new abstract transformer is rated by its precision on parts of the input space not covered by transformers that have previously been synthesized.
This adaptive strategy steers the synthesis process away from inputs that are already handled precisely and toward those that require new cases.
This decomposition not only makes it easier to synthesize precise transformers with high precision, but also allows users to control the number and size of components in the meet---providing an easy knob for tuning efficiency vs.\ precision.

\paragraph{\changed{Pattern 2:} Separate transformers for separate jobs}
Beyond splitting up the input space, we noticed that realistic transfer functions gain precision by exploiting information that is present in the IR\@.
For example, the LLVM compiler's abstract transformer for integer multiply begins with a large special case for the ``nsw'' or ``no unsigned wrap'' flag\footnote{\url{https://github.com/llvm/llvm-project/blob/release/20.x/llvm/lib/Analysis/ValueTracking.cpp\#L383-L437}}---this flag can be exploited to increase analysis precision, because signed overflows become undefined.
Then, immediately inside the nsw case, the code splits again to handle the case of computing the square of a value, which again affords additional precision.
Our observation is that it is unnecessary and even undesirable to entangle the implementation of the $y * y$ case with the $x * y$ case and the $x\ {*}{_{\footnotesize\mbox{nsw}}}\ y$ case: these are actually distinct transfer functions that need to be handled separately during testing or formal verification.
They happen to be handled by overlapping code only because LLVM's developers decomposed their code that way.
Program synthesis, on the other hand, changes the basic software engineering economics, making it cheap to create a large number of abstract transformers, including a specialized version for every IR-level condition that \benchanged{can} be exploited to increase precision.

\paragraph{Implementation in MLIR}
We get formal semantics for instructions from previous work on the SMT dialect for MLIR~\cite{mlir-pldi25}, which supports lowering operations to both LLVM IR and SMT formulae.
The LLVM IR lowering, and subsequent JIT compilation, enables fast evaluation of candidate transformers during synthesis, and allow us to leverage LLVM's powerful optimizers to improve the performance of the final synthesized transformers.
The SMT lowering allows us to verify the soundness and precision of the synthesized transformers, although in practice we generally measure precision via testing rather than solving because we are interested in giving our stochastic optimizer a hill to climb, rather than a binary result (a model counting solver would be an alternative way to get a hill to climb, but in our experience they do not scale to jobs like this one).

\paragraph{Evaluation}

We evaluate our approach by synthesizing abstract transformers for the abstract domains and operations used in the LLVM IR. The results show that \name complements the precision of LLVM transformers (when measured on 8-bit and 64-bit integers) of \benchanged{7}/\benchanged{47} transformers in the \knownbits domain and \benchanged{19}/\benchanged{47} in \constantrange.
With the addition of a handwritten reduced-product operator for combining synthesized transformers across different abstract domains, our synthesized transformers exceed LLVM's precision on 22/\benchanged{47} operators.

\paragraph{Contributions} Our work makes the following contributions:
\begin{itemize}
    \item We propose a framework for synthesizing abstract transformers that leverages existing formal semantics for instructions, and is not limited to specific abstract domains and does not require program templates (\Cref{se:problem-definition}).
    \item We design an algorithm that incrementally synthesizes the meet of multiple abstract transformers (\Cref{se:framework}), which enables an MCMC-based search procedure that can discover individual smaller transformers that can be added to the meet (\Cref{se:algorithm}).
    \item We implement the algorithm in \name, a tool that effectively balances MCMC-based exploration with SMT-based verification.
    We apply \name on the bread-and-butter abstract domains from LLVM, namely \knownbits and \constantrange~(\Cref{sec:implementation}).
    \item We conduct an evaluation showing how \name can synthesize abstract transformers for real LLVM operators.
    \changed{Our transformers are often complementary to LLVM's,
    and in \benchanged{some} cases
    exceed the precision of LLVM's hand-tuned transformers~(\Cref{se:evaluation}).}
\end{itemize}





%% file: 2illustration_example.tex
\section{Problem Definition and Overview of the Approach}
\label{se:problem-definition}

In this section, we define the problem addressed by our framework using a toy example (\Cref{sec:problem-def-subsec}).
We then provide an excerpt of a real transformer from MLIR, \texttt{urem} over known bits, to illustrate how our problem setting and approach leads to practical gains (\Cref{sec:highlight}).

\subsection{The Transformer Synthesis Problem}
\label{sec:problem-def-subsec}
Throughout this section, we use a running example in which the goal is to synthesize transformers for the integer maximum function $f(x, y) = \max(x,y)$ over the domain of intervals.

The user of our framework needs to provide a definition of a concrete domain, an abstract domain over which they are trying to synthesize abstract transformers, and a language of programming constructs the synthesizer can use to synthesize the abstract transformers.

\mypar{Concrete Domain and Concrete Transformers}
A transformer depends on a concrete domain $\cdom$, and a concrete transformer $f: \cdom^k \to \cdom$.
In our example, we let $\cdom$ be integers of bitwidth up to a certain number (e.g., 32), which can be represented by the \apint class in LLVM and MLIR.\footnote{
Note that the set of integers of bitwidths up to $w$ is \textbf{not} equivalent to the set of integers in the range $[0, 2^w - 1]$.
Instead, it is the union of sets of bitvectors with lengths from $1$ to $w$. In other words, $\cdom = \cdom_{1} \cup \cdom_{2}\cup \cdots \cup \cdomp{w}$, where each $\cdomp{i}$ is a concrete subdomain representing integers with bitwidth $w$, and these subdomains are disjoint. 
We always assume the concrete transformer is only defined on inputs within the same subdomain $\cdomp{i}$.}
%
%
The concrete transformer in our example is $\max: \cdom^2 \to \cdom := \lambda x \lambda y. \texttt{ite}(x>y,x,y)$.

%

\mypar{Abstract Domain} 
\label{se:problem-definition:abstract-domain}
A lattice $\adom$ serves as an abstract domain.
%
%
In this subsection, $\adom$ is the interval domain. Each abstract value $a$ is a pair of integers representing an interval $[a.l, a.r]$.\footnote{\changed{As we only consider bounded integers of specific bitwidths, the abstract domain is restricted to closed intervals and is therefore finite.} Readers may also notice that $a.l$ and $a.r$ should have the same bitwidth as the concrete values included by that interval.
Strictly speaking, given the full concrete domain $\cdom = \cdom_{1} \cup \cdom_{2}\cup \cdots \cup \cdomp{w}$, each concrete subdomain $\cdomp{i}$ has its own abstract domain $\adomp{i}$, each equipped with its own top and bottom elements.
The full abstract domain $\adom$ is likewise the disjoint union of all abstract subdomains $\adomp{i}$.
The top element of $\adom$ will only be reached when joining two abstract values from two different $\adomp{i}$, and the bottom element will only be reached when meeting two abstract values from two different $\adomp{i}$.
However, in practice, the abstract transformers only need to be defined on inputs from the same $\adomp{i}$.}

In particular, the user needs to provide the implementations of the following components:
\begin{itemize}
    \item Concretization function $\gamma: \adom \to 2^{\cdom}$. In our example, $\gamma(a) = \{a.l, a.l+1\cdots{}, a.r\}$.
    \item Meet function $\sqcap: \adom \times \adom \to \adom$. In our example, $\sqcap$ is the intersection of two intervals, i.e., \changed{$a \sqcap b = \textbf{if}\; \max(a.l, b.l) > \min(a.r, b.r) \;\textbf{then}\; \bot \;\textbf{else}\; [\max(a.l, b.l), \min(a.r, b.r)]$}.
    \item Join function $\sqcup: \adom \times \adom \to \adom$. In our example, $\sqcup$ is the union of two intervals, i.e., $a \sqcup b = [\min(a.l, b.l), \max(a.r, b.r)]$.
    \item Single value abstraction function $\beta: \cdom \to \adom$.  In our example, $\beta(x) = [x,x]$.
\end{itemize}

The general abstraction function $\alpha: 2^{\cdom} \to \adom$ that maps a set of concrete elements to their abstract one is defined as $\alpha(C) := \bigsqcup_{x \in C} \beta(x)$.
The partial order $\sqsubseteq \in \adom \times \adom$ that relates abstract elements is defined as follows: $a_1 \sqsubseteq a_2 \iff a_1 \sqcap a_2 = a_1$.

\mypar{Language} 
In our setting, a domain-specific language (DSL) $\lang$ is a context-free grammar of the form $E := a \mid c \mid op(E_1, \cdots, E_k)$, where $a$ ranges over abstract input variables, $c$ denotes constants drawn from a fixed set $C$, and $op$ is an operator drawn from a predefined set of function symbols $\langop$ supported by the DSL.
In our example, $\langop\cup C$ consists of the set 
$\{+, -, \&, |, \min, \max, [\cdot,\cdot], \cdot.l, \cdot.r, \texttt{Zero}, \texttt{AllOnes}\}$, 
where $[\cdot, \cdot]$ denotes interval construction, and $\cdot.l$ and $\cdot.r$ access the left and right endpoints of an interval, respectively. The \texttt{Zero} and \texttt{AllOnes} constants return all zeros and all ones at the given bitwidth.
For instance, the function 
$\tf(a, b) := [\min(a.l, \texttt{AllOnes}),\ \max(a.r, \texttt{Zero})]$ 
is a valid program in the language $\lang$.

\mypar{Soundness and Precision}
Our goal is to synthesize a set of \emph{sound} (i.e., valid overapproximations of the function behavior)
and \emph{precise} (tight) abstract transformers $\settf = \{\tf_1, \tf_2, \cdots, \tf_n\}$ expressed using DSL operators.
%
Because the user of the framework provides the meet operation for the abstract domain as input, we can then compute the meet of all such transformers $\meettf$ as follows:

\changed{\begin{definition}[Meet of Transformers]
\label{def:meet-tf}
Given two abstract transformers $\tf_1, \tf_2 : \adom^k \to \adom$, 
their \emph{meet} is defined as the transformer $\tf_1 \sqcap \tf_2 : \adom^k \to \adom$ such that, for all $a_1, \dots, a_k \in \adom$:
$(\tf_1 \sqcap \tf_2)(a_1, \dots, a_k) 
\;=\; \tf_1(a_1, \dots, a_k) \;\sqcap\; \tf_2(a_1, \dots, a_k)$.
We define the meet of a set of transformers $\settf \subseteq (\adom^k \to \adom)$ as $\meettf \;=\; \bigsqcap_{\tf \in \settf} \tf$.
\end{definition}}

\changed{Intuitively, the meet of two transformers is their pointwise meet in the abstract domain, which is both sound and represents the most precise possible combination of the synthesized transformers.}\footnote{Ideally, the set of abstract transformers would be a singleton $\{\besttf\}$.
However, the theoretical best transformer $\besttf$ might not be expressible in the DSL, or may be too complex and thus computationally expensive for static analysis.}

One can check that a transformer is sound using the concretization function $\gamma$ as follows:

\begin{definition}[Soundness of Transformers]
\label{def:soundness}
A transformer $\tf: \adom^k \to \adom$ is \emph{sound} with respect to a concrete function $f: \cdom^k \to \cdom$, \changed{denoted by $\issound(\tf)$, if it is sound on all abstract inputs}, i.e.,
\changed{
\[
\issound(\tf) \eqdef
\forall a_1, \dots, a_k \in \adom.\quad 
\{ f(c_1,\dots,c_k) \mid c_i \in \gamma(a_i) \}
\;\subseteq\; 
\gamma(\tf(a_1,\dots,a_k)).
\]
}
\end{definition}
For precision, there is no easy way to check that a transformer is the most precise possible among those expressible in a given language (a problem as hard as checking unrealizability in program synthesis~\cite{hu2019unrealizability}).
\changed{We address this practical problem by introducing a precision measure; namely,}
%
we ask the user to provide a norm function $\norm{\cdot}: (\adom^k \to \adom) \to \mathbb{N}$ that quantifies the imprecision of a transformer.
%
The objective of our problem is to minimize the norm of the synthesized transformer $\meettf$.
%
%
While a norm function over transformers can be hard to define, 
in our implementation, we ask the user to provide a size function $\size{\cdot}: \adom \to \mathbb{N}$ on abstract values,
and derive the norm via the size function as $\norm{\tf} = \sum_{a\in \adom} \size{\tf(a)}$.
%
In practice, the size $\size{a}$ can be set as any function that is monotonic with respect to the actual size of the concretization set $\gamma(a)$.
\changed{For the interval domain, we define the size as the $\log_2$ of its length, e.g., $\size{[0,7]} = \log_2(8) = 3$.}
%

Because computing the sum of norms over the entire abstract domain can be expensive, we will often approximate the size by evaluating it over a representative subset $\adom' \subseteq \adom$. We use $\normp{\tf}{\adom'} = \sum_{a\in \adom'} \size{\tf(a)}$ to denote the approximate norm over the subdomain $\adom'$.
In our example, the subset $\adom'$ can be all abstract values represented by integers up to a smaller bitwidth, or some sampled abstract values represented by integers at a large bitwidth.

\paragraph{\changed{Norms vs. Metrics}}
\changed{Initially, we formalized precision using distance metrics rather than norms.
This approach turned out to be problematic because one would need to compute the distance between a synthesized transformer $f_{synth}^\sharp$ and the theoretically best transformer $\besttf$ by enumerating all abstract values and their concretizations, and computing the join of $\beta(f(c))$ for each concrete value.
This is feasible only at small bitwidths.
That being said, prior works use metrics~\cite{logozzo2009measure,Casso2020Measure,Campion2023Measure}
and may provide a way
to explore non-integer or infinite abstract domains in the future.
}

\subsubsection{Problem Definition}
%
We are now ready to define the problem solved in this paper:

\begin{definition}[Transformer Synthesis Problem]
\label{def:prob-def}
    Given a concrete transformer $f: \cdom^k\to \cdom$, an abstract domain 
    $(\adom, \top, \gamma, \sqcap, \sqcup, \beta)$, a norm function $\norm{\cdot}:  (\adom^k \to \adom) \to \mathbb{N}$, and a DSL $\lang$, the \textit{transformer synthesis problem} is to find a set of transformers $\settf = \{\tf_1, \tf_2, \cdots, \tf_n\}$ in $\lang$ such that
    \begin{itemize}
        \item Their meet \changed{$\meettf$} is sound: $\issound(\meettf)$.
        \item The norm of \changed{$\meettf$} is minimal\benchanged{, i.e., there is no sound set of transformers $\settfp{G}$
        such that $\norm{\meettfp{G}} < \norm{\meettf}$.}
              
        \item No $\tf_i \in \settf$ is redundant:
              $\forall \tf_i\in \settf, \exists \vec{a} \in \adom^k,   \left(\bigsqcap_{\tf \in {\settf \setminus \{\tf_i\}}}\tf(\vec{a})\right) \not\sqsubseteq \tf_i(\vec a)$.
    \end{itemize}
\end{definition}
The first requirement, that the meet of the $n$ transformers is sound, is satisfied if all $n$ transformers are sound.
The second requirement asks that the final meet be as precise as possible, rather than requiring each individual transformer to be.
The last requirement ensures every transformer in the final solution contributes to improving the precision.\footnote{\citet{DBLP:journals/pacmpl/KalitaMDRR22} propose a similar definition for the problem of synthesizing \textit{one} most precise abstract transformer in a given DSL.
In their setting, precision is defined in absolute terms with respect to the $\sqsubseteq$ partial order.
In our setting, precision is defined in terms of a size function over the abstract domain.
Furthermore, our definition extends to the set of synthesizing multiple incomparable transformers.
We further discuss these implications in \Cref{se:related-work}.
}

Returning to our running example $f(x, y) = \max(x, y)$, \Cref{fig:toy-example} shows the problem inputs and one possible set of output transformers.
Smart readers may observe that the best transformer has a succinct representation: $\tf(a, b) = [\max(a.l, b.l), \max(a.r, b.r)]$.
However, the meet of the four output transformers in \Cref{fig:toy-example} is also equivalent to this best transformer—since the maximum of their left endpoints equals $\max(a.l, b.l)$ and the minimum of their right endpoints equals $\max(a.r, b.r)$.
Although some individual transformers contain ``unnecessary fragments'' such as $a.l \& b.r$, each transformer still has a smaller size than the best transformer as a single monolithic expression.

\input{figures/toy-example}

\subsection{Case Study: Synthesizing a Precise Transformer for \texttt{urem}}
\label{sec:highlight}

To illustrate the practical capabilities of \name, we present a detailed case study drawn from our evaluation: synthesizing a \knownbits transformer for the \texttt{urem} (unsigned remainder) operation. This example demonstrates how \name synthesizes precise and non-trivial transformers that both match and exceed hand-written implementations in LLVM.

\subsubsection{Background: The\/ \knownbits Domain and the\/ \texttt{urem} Operator}

The \knownbits abstract domain models partial bit-level knowledge of integer values using two disjoint bitvectors: \knownzero and \knownone. A bit is definitely zero if set in \knownzero, definitely one if set in \knownone, and unknown if unset in both. This domain is widely used in compiler optimization passes such as those in LLVM.

We focus on the unsigned remainder operation \texttt{urem}, which computes $L \bmod R$ for unsigned integers $L$ (dividend) and $R$ (divisor). Optimizing the transformer for \texttt{urem} is challenging due to its non-linear behavior and the many edge cases involving known bits.

\subsubsection{Reference Implementation: LLVM's Hand-Written Transformer}

\input{codes/llvm-urem}

LLVM includes a hand-written \knownbits transformer for \texttt{urem}, shown in \Cref{fig:llvm-urem}. It is implemented using utility functions from the \texttt{APInt} library, such as:
\rone \texttt{countTrailingZero(x)} and \texttt{countLeadingZero(x)}: compute the number of trailing or leading zeros, and \rtwo \texttt{setHighBits(x, k)} and \texttt{setLowBits(x, k)}: construct bitvectors with the highest or lowest $k$ bits set to 1.

This transformer ``applies'' three heuristics:
\rone If the divisor $R$ is a multiple of $2^k$ (i.e., it has $k$ trailing zeroes), then the lowest $k$ bits of the result are equal to those of $L$.
\rtwo If the divisor \texttt{R} is known to be a constant, some special-case handling is applied.
\rthree Since $L \bmod R < \min(L, R)$, the number of leading zeros in the result must be at least as large as in both $L$ and $R$.

\subsubsection{Synthesized Transformer in MLIR: Matching and Extending LLVM}

\Cref{fig:mlir-urem} shows the transformer synthesized by \name in MLIR form. The solution \texttt{@solution} computes the meet of nine independently synthesized transformer candidates \texttt{@f1} through \texttt{@f9}.

Upon manual inspection, we find that 8/9 synthesized components recover key heuristics from the LLVM implementation. 
The remaining one \benchanged{is a new heuristic}.

\mypar{Recovering Existing Heuristics}
As an example, transformer \texttt{@f1} matches LLVM’s third heuristic. It computes the number of leading zeros in the dividend \texttt{\%L}, uses that to set the highest bits in a bitvector, and constructs a \knownbits value accordingly. This sound heuristic (also used by LLVM) encodes the fact that the result of \texttt{urem} must have at least as many leading zeros as the dividend.

\mypar{Discovering New Heuristics}
Transformer \texttt{@f2} illustrates synthesis that eludes human intuition. 
It defines a condition \texttt{@f2\_cond} and a guarded body \texttt{@f2\_body}, returning \texttt{ite(@f2\_cond, @f2\_body, \%top)} so the body applies only when the condition holds.
The body \texttt{@f2\_body} \benchanged{returns the dividend as the result.
This is unsound in general, however,} the synthesizer simultaneously generates a guard \texttt{@f2\_cond} that ensures soundness: the transformer is only used when the dividend's maximum possible value is less than the divisor's minimum possible value. In that case, the remainder equals the dividend, and the transformer is sound.

This heuristic---``if the dividend is provably less than the divisor, then \texttt{urem(L, R)} equals \texttt{L}''---is absent from the LLVM transformer, highlighting the power of synthesis in discovering useful but overlooked cases. That synthesis outputs the code realizing this case is the cherry on top.

Our synthesized transformer is more precise than the LLVM implementation---according to our precision metric---but also complementary. That is, the LLVM and synthesized transformers are incomparable: neither subsumes the other. Their meet yields a strictly more precise transformer.

This case study highlights three key outcomes:
\begin{enumerate}
\item \name recovers hand-crafted compiler heuristics automatically.
\item \name discovers new, sound heuristics that are absent in existing implementations.
\item The transformers synthesized by \name can offer strictly better precision when combined with human-crafted ones.
\end{enumerate}

This example demonstrates how \name can serve as a practical, drop-in synthesis engine for compiler frameworks such as MLIR, producing transformers that are not only correct and efficient but also competitive with and complementary to those written by domain experts.

\input{codes/mlir-urem}

%% file: figures/toy-example.tex
\begin{figure}[t]
\centering
\begin{minipage}[t]{0.6\textwidth}
\begin{tcolorbox}[colback=gray!5!white, colframe=gray!75!black, title=Input]
\vspace{-6pt}
\begin{align*}
\text{Concretization:} \quad & \gamma([a.l, a.r]) = \{a.l, \cdots, a.r\} \\
\text{Meet:} \quad & a \sqcap b = [\max(a.l, b.l), \min(a.r, b.r)] \\
\text{Join:} \quad & a \sqcup b = [\min(a.l, b.l), \max(a.r, b.r)] \\
\text{Abstraction:} \quad & \beta(x) = [x, x] \\
\text{Concrete op:} \quad & f(x, y) = \max(x, y) \\
\text{DSL ops:} \quad & \{+, -, \&, |, \min, \max, \cdots\} \\
\text{Size:} \quad & \size{a} = \lfloor\log_2(\abs{a.l - a.r})\rfloor
\end{align*}
\end{tcolorbox}
\end{minipage}
\begin{minipage}[t]{0.38\textwidth}
\begin{tcolorbox}[colback=gray!5!white, colframe=green!25!black!75, title={Output \changed{(\name)}}]
\vspace{-11pt}
\begin{align*}
\\
\tf_1(a, b) &= [\texttt{Zero}, \max(a.r, b.r)] \\
\tf_2(a, b) &= [a.l~\&~b.l, \texttt{AllOnes}] \\
\tf_3(a, b) &= [a.l, a.l \mid b.l] \\
\tf_4(a, b) &= [b.l, \texttt{AllOnes}]\\
\\
\end{align*}
\end{tcolorbox}
\end{minipage}
\caption{Input and output for the transformer synthesis problem on a toy example.}
\label{fig:toy-example}
\end{figure}

%% file: codes/llvm-urem.tex
\begin{figure}
\begin{lstlisting}[language=llvm, numbers=none]
KnownBits urem(KnownBits L, KnownBits R) {  
  // Part 1:
  unsigned RTrailingZeros = R.Zero.countTrailingZero();
  APInt Mask = setLowBits(0, RTrailingZero);
  APInt knownZero = L.Zero & Mask;
  APInt knownOne = L.One & Mask;
  // Part 2:
  if (R.isConstant() && R.getConstant().isPowerOf2()) { ... }
  // Part 3:
  unsigned Leaders = max(L.Zero.countLeadingZero(), R.Zero.countLeadingZero());
  knownZero = setHighBits(knownZero, Leaders);
  return {knownZero, knownOne};
}
\end{lstlisting}
\caption{The \knownbits transformers for \texttt{urem} operator in LLVM}
\label{fig:llvm-urem}
\end{figure}

%% file: codes/mlir-urem.tex
\begin{figure}
\begin{lstlisting}[language=MLIR, numbers=none]
func.func @f1(%L : KnownBits, %R : KnownBits) -> KnownBits {
    %1 = countLeadingZero(%L.zero) 
    %knownZero = setHighBits(0, %1)
    return makeKnownBits(%knownZero, 0)
}
func.func @f2_cond(%L: KnownBits, %R: KnownBits) -> bool {
    %Lmax = negate(%L.zero)
    %Rmin = %R.one
    %cond = unsignedLessThan(%Lmax, %Rmin)
    return %cond
}
func.func @f2_body(%L : KnownBits, %R : KnownBits) -> KnownBits {
    return %L
}
func.func @f2(%L : KnownBits, %R : KnownBits) -> KnownBits {
    return ite(@f2_cond(%L, %R), @f2_body(%L, %R), %top)
}
...
func.func @solution(%L : KnownBits, %R : KnownBits) -> KnownBits {
    return meet(@f1(%L, %R), ... , @f9(%L, %R))
}
\end{lstlisting}
\caption{Our synthesized \knownbits transformers for \texttt{urem}, written in MLIR.}
\label{fig:mlir-urem}
\end{figure}

%% file: 3framework.tex
\section{An Ideal Algorithm for the Transformer Synthesis Problem}
\label{se:framework}

In this section, we present an idealized synthesis algorithm for solving the transformer synthesis problem~(\Cref{def:prob-def}).
We will provide a practical instantiation using MCMC search in \Cref{se:algorithm}.

Our algorithm draws inspiration from recent approaches for synthesizing the most precise conjunctive specifications~\cite{DBLP:journals/pacmpl/ParkDR23,loudfull}. Rather than generating an entire conjunction in one step, these methods iteratively synthesize individual conjuncts, ensuring that each new conjunct strictly improves the overall precision.
We adopt a similar strategy, tailored to the transformer setting, by incrementally synthesizing individual transformers that refine precision.

\subsection{Synthesizing One Transformer at a Time}
In our setting, meet operations over transformers play the role of conjunctions, and individual transformers correspond to conjuncts. 
\Cref{alg:idealized-algo} maintains a set of sound, incomparable transformers $\soundset$, initialized to the empty set, representing the most imprecise transformer $\top$. 
The algorithm iteratively synthesizes new transformers that, when combined via meet, minimize the imprecision measured by the norm function.
It loops as long as it can find a precision improvement:

\begin{algorithm}[htbp]
    \footnotesize
    {\it
        \caption{IdealSynthesizeTransformers(\problem) \label{alg:idealized-algo}}
        \DontPrintSemicolon
        
        \textbf{Input:} $\problem$ — An instance of the Transformer Synthesis Problem. \\
        \textbf{Output:} $\soundset$ — A set of synthesized transformers solving $\problem$. \\
    
        \SetKw{Break}{Break}
        $\soundset \gets \varnothing$\tcp*{Initialize to most imprecise transformer set}
        
        \While{true}{         
            $f \gets \changed{\findnewtf}(\soundset, \problem)$ \tcp*{Synthesize transformer maximizing precision gain}
            \If{$\norm{f\sqcap \soundsetmeet }= \norm{\soundsetmeet}$}{
                \Return \texttt{RemoveRedundant}($\soundset$) \tcp*{Termination: remove transformers, preserving norm}
                \tcp*{(compute a set greedily; it may not be unique)}
            }
            $\soundset \gets \soundset \cup \{f\}$ \tcp*{update set of synthesized transformers}
        }        
    }
\end{algorithm}


A key advantage of this iterative approach is the reduction of the synthesis problem to repeatedly generating individual transformers. Specifically, rather than directly synthesizing a full set $\soundset$ minimizing the norm $\norm{\soundsetmeet}$, each iteration seeks a transformer $f$ that minimizes the norm when added to the set of transformers $\soundset$ we already have synthesized:

\changed{
\begin{equation}
\label{eq:one-best-condition}
\underset{f \in \lang,\; \issound(f)}{\operatorname{minimize}}
\;\; \norm{f \sqcap \soundsetmeet}.
\end{equation}}

This structure simplifies the search space and modularizes synthesis, as each $\findnewtf(\soundset, \problem)$ call focuses solely on the next incremental improvement.

\begin{theorem}[Soundness]
\label{thm:soundness}
If \findnewtf synthesizes a single sound transformer that is a solution to \Cref{eq:one-best-condition} and there exists a finite solution to the transformer synthesis problem, then \Cref{alg:idealized-algo} returns a solution to the transformer synthesis problem.
\end{theorem}
\begin{proof}
Let $\settf$ be a solution to the synthesis problem such that for every $\norm{\settf}=k$.
Because every call to \findnewtf reduces the norm (and the norm is an integer), the algorithm terminates.
Furthermore, the final set satisfies all three conditions: all transformers are sound, the precision cannot be improved further, and no transformer is redundant.
\end{proof}
\changed{Note that any intermediate result of the algorithm is
a valid transformer, though (probably) sub-optimal}.
This property is important for our MCMC-based approach in \Cref{se:algorithm}.

\subsection{Focusing Precision on Relevant Inputs}
\label{sec:focusing-precision}

A further benefit of this approach is the ability to focus synthesis efforts on the inputs that matter, that is, inputs where the current transformer set $\soundset$ is still imprecise.

While soundness requires that synthesized transformers behave correctly on all inputs, precision only needs to improve where existing transformers leave room for refinement. Thus, each $\findnewtf$ call can restrict its optimization to the subset of ``imprecise'' inputs:
\begin{equation}
\adom_\text{imprecise} = \adom \setminus \{ a \mid \soundsetmeet(a) = \besttf(a)) \}.
\end{equation}
We can equivalently rewrite the objective for the next transformer $f$ as:

\changed{
\begin{equation}
\underset{f \in \lang}{\operatorname{minimize}}
\;\; \quad \normp{f \sqcap \soundsetmeet}{\adom_\text{imprecise}}
\end{equation}}

Here, $\normp{\cdot}{\adom_\text{imprecise}}$ computes norm considering only inputs in $\adom_\text{imprecise}$. This refinement preserves correctness while avoiding wasted effort on already-precise regions of the input space.
Formally,
\changed{
$\underset{f \in \lang}{\operatorname{minimize}}\quad  \norm{f \sqcap \soundsetmeet}
\iff
\underset{f \in \lang}{\operatorname{minimize}}\quad \normp{f \sqcap \soundsetmeet}{\adom_\text{imprecise}}$.
}

This selective focus is efficient (it allowing us to compute the norm on fewer inputs) and also aligns with standard optimization principles: prioritize areas of maximum potential gain.

%% file: 4algorithm.tex
\section{Randomly Searching for Abstract Transformers using MCMC}
\label{se:algorithm}

The ideal algorithm presented in \Cref{se:framework} incrementally synthesizes abstract transformers that collectively \changed{form the solution to the transformer synthesis problem (\Cref{def:prob-def})}. To make this process practical, we implement the core routine $\findnewtf(\soundset, \problem)$ using a stochastic search procedure $\texttt{MCMCSynthesizeTransformer}(\soundset, \problem)$ based on Markov Chain Monte Carlo (MCMC).
\changed{Our approach is inspired by Stoke~\cite{Alex2013Stoke} and contributes a novel cost function and abductive refinement strategy.}

\begin{algorithm}[htbp]
    \footnotesize
    {\it
        \caption{$\texttt{\changed{MCMCSynthesizeTransformer}}(\soundset, \problem)$}
        \label{alg:mcmc-best-transformer}
        \DontPrintSemicolon

        \textbf{Input:} 
        $\soundset$ — Current set of synthesized transformers; \\         
        \textbf{Output:} A new sound transformer $f \in \lang$ that (in the limit) minimizes precision. \\


        \SetKwProg{Fn}{fun}{:}{}
        \SetKw{None}{None}
        
        \Fn{$\costfunc(\tf)$ \label{alg:mcmc-best-transformer-cost}}{
        \Return $\lambda(1{-}\sound(\tf)) + \kappa(1 {-} \improvement(\tf, \soundset))$ 
        \tcp*{reward soundness, precision improv.}
    }

    $f \gets \texttt{initialize}()$ \tcp*{random initial program} \label{alg:mcmc-best-transformer-init}

    \For{$i \gets 1$ \KwTo $\nstep$}{ \label{alg:mcmc-best-transformer-loop}
        $f' \gets \texttt{mutate}(f)$ \tcp*{mutate current candidate} \label{alg:mcmc-best-transformer-mutate}

        $p \sim \mathcal{U}(0, 1)$ \tcp*{sample acceptance threshold}

        \If {$\costfunc(f) - \costfunc(f') > \mcmctemp \cdot \log(p)$}{ \label{alg:mcmc-best-transformer-accept}
            $f \gets f'$ \tcp*{accept proposed candidate}
        }
    }
    
    \If {$\sound(f) < 1$}{ \label{alg:mcmc-best-transformer-fallback}
        \Return $\top$ \tcp*{return trivial top transformer if no sound one found}
    }
    
    \Return $f$ \tcp*{return the lowest cost sound transformer found}
}
\end{algorithm}

The goal of this random search, presented in \Cref{alg:mcmc-best-transformer}, is to synthesize a transformer that minimizes the cost function defined in Line~\ref{alg:mcmc-best-transformer-cost}. The cost combines two objectives: maximizing soundness and minimizing norm. 
%
$\sound(f)$ returns a number between 0 and 1 representing the fraction of inputs on which the transformer $f$ is sound.
$\improvement(f, g)$ returns a number between 0 and 1 representing how much $f \sqcap g$ improves the precision of $g$.
Formally, we define a predicate $\soundat(f, a) := \besttf(a) \sqsubseteq f(a)$ indicating transformer $f$ is sound at an abstract input $a$. As we target only finite abstract domains, \sound and \improvement can be defined as in \Cref{eq:soundness} and \Cref{eq:improvement}, respectively:
\begin{equation}
\label{eq:soundness}
\sound(f) := \left(\sum_{a \in \adom} \mathbf{1}[\soundat(f, a)]\right)/|\adom|
\end{equation}

\begin{equation}
\label{eq:improvement}
    \improvement(f, g) := \left(\sum_{a \in \adom} \mathbf{1}[\soundat(f, a)] \cdot \left(\size{g(a)} -  \size{f(a) \sqcap g(a)}\right) \right) / \normp{g}{\adom}
\end{equation}

Note that \improvement counts precision gain only on sound inputs.
If $f$ is sound over the entire abstract domain, we have $\improvement(f, g) = 1 - \normp{f\sqcap g}{\adom}/\normp{g}{\adom}$
because $\Sigma_{a \in \adom} \size{g(a)} = \normp{g}{\adom}$.

The algorithm begins by sampling a random candidate program $f$ from the DSL $\lang$ (Line~\ref{alg:mcmc-best-transformer-init}). It then performs $\nstep$ iterations of local search, where in each iteration a syntactic mutation produces a new candidate $f'$ (Line~\ref{alg:mcmc-best-transformer-mutate}). If $f'$ has a lower cost than the current candidate, it is accepted; otherwise, it is accepted with a probability determined by the difference in cost and temperature parameter $\mcmctemp$, following the standard Metropolis-Hastings acceptance rule (Line~\ref{alg:mcmc-best-transformer-accept}).

If no sound transformer is found after the search, the algorithm returns the trivial transformer $\top$, representing the most imprecise but sound abstraction (Line~\ref{alg:mcmc-best-transformer-fallback}). Otherwise, the \changed{most precise} sound transformer discovered is returned.

The use of MCMC for synthesizing transformers provides important asymptotic guarantees'
\changed{assuming ergodicity (discussed below)}.
In the limit, the search procedure samples transformers according to a distribution biased toward lower-cost candidates. 
As the number of iterations tends to infinity, the probability of synthesizing a sound transformer that is sound and minimizes norm approaches one.
The following corollary of the standard convergence guarantees for Metropolis-Hastings MCMC~\cite{hastings1970} captures this result:

\begin{corollary}[Asymptotic Optimality of MCMC Search]
\label{thm:mcmc-optimality}
Assume the proposal distribution used by $\texttt{mutate}(\cdot)$ is ergodic over the space of programs expressible in $\lang$.
Then, as $\nstep \to \infty$, with probability approaching 1, $\texttt{MCMCSynthesizeTransformer}(\soundset, \problem)$ returns $f$ satisfying:
\[
\norm{f \sqcap \soundsetmeet} = \min_{\substack{f' \in \lang \\ \sound(f') = 1}} \norm{f' \sqcap \soundsetmeet}.
\]
\end{corollary}

\changed{
In short, repeated or sufficiently long runs of  $\texttt{MCMCSynthesizeTransformer}$ should yield transformers that are sound and highly precise in terms of the norm function.
Though each invocation of the algorithm is approximate, in expectation the overall result tends to maximum precision.}

\paragraph{\changed{On Ergodicity.}}
\changed{Our setup is designed to provide an ergodic search space.
First, the space is in a simple format parameterized by the operation set,
and further constrained by the program representation:
each operation takes two abstract inputs, consists of a sequence of SSA instructions, and produces one abstract output.
Only the instructions in the middle are subject to mutation.
Second, our mutation strategy ensures the ergodicity of the Markov chain because the two possible mutations are invertible.
Thus, there is a positive probability of transitioning from any program to any other.
Overall, our strategy is similar to the one employed by Stoke~\cite{Alex2013Stoke}.
}

\subsection{Two Strategies for Randomly Sampling Programs}

\begin{algorithm}[tbp]
\footnotesize
{\it
\caption{\benchanged{randomly initializing and randomly mutating transformer operations}}
\label{alg:mcmc-normal}
\DontPrintSemicolon

\SetKwProg{Fn}{fun}{:}{}
\SetKw{None}{None}

\Fn{$\texttt{initializeNormal}()$ \label{alg:mcmc-normal-init}}{
$l\gets$ maximum number of lines  

\For{$i \gets 1$ \KwTo $l$}{
    $op\gets $ \texttt{sample}(operator in $\lang$)

    $a\gets \texttt{sample}(0,i-1)$ \tcp*{only use earlier variable indices or the input variable $x_0$} 

    $b\gets \texttt{sample}(0,i-1)$ 

    $s[i] \gets (x_i = op(x_a, x_b))$
}

\Return $s[1] \cdots s[l]$ 
}
\Fn{$\texttt{mutateNormal}(f)$ \label{alg:mcmc-normal-mutate}}{
$x_i = op(x_a, x_b) \gets$ \texttt{sample}(line in $f$) \tcp*{sample a line of the program to mutate}

$op'\gets $ \texttt{sample}(operator in $\lang$)

$a'\gets \texttt{sample}(0,i-1)$ 

$b'\gets \texttt{sample}(0,i-1)$ 

\Return $f[x_i = op(x_a, x_b)/x_i = op'(x_{a'}, x_{b'})]$ \tcp*{replace $i$-th line with randomly sampled one} 
}
}
\end{algorithm}

We now describe two complementary strategies that we can \benchanged{use in}~\Cref{alg:mcmc-best-transformer} to synthesize transformers via random search. 
The strategies correspond to different implementations of the $\texttt{initialize}$
and $\texttt{mutate}$ procedures in \Cref{alg:mcmc-best-transformer}~(Lines~\ref{alg:mcmc-best-transformer-init} and~\ref{alg:mcmc-best-transformer-mutate}).

\subsubsection{Randomly Mutating Transformer Operations}
\label{sec:randomly-mutating}

The first strategy performs random syntactic mutations to the current transformer by altering its operations or structure (\Cref{alg:mcmc-normal}). 
This corresponds to traditional MCMC-based synthesis, which operates directly over the DSL $\lang$~\cite{Alex2013Stoke}.

For simplicity, the formalization assumes all operations are binary and that $x_0$ is the input variable to the transformer.
In practice there are more possible inputs as well as constants that can be used in the transformer.
\Cref{alg:mcmc-normal} defines two core functions: \texttt{initializeNormal} and \texttt{mutateNormal}, which build and mutate transformers as sequences of statements.
The initialization function creates a transformer with $l$ statements (this is a parameter in our implementation).
Each statement $s[i]$ assigns a variable $x_i$ by applying a randomly selected operator from $\lang$ to two variables previously defined variables at lower indices. This ensures well-formed data dependencies.

The mutation function (\texttt{mutateNormal}) selects a random statement, sampled using a user-given probability distribution, and replaces it with a newly sampled statement, again selecting operator and operands consistent with variable dependencies.

This setup supports incremental local modifications during MCMC sampling, enabling the search to explore transformer candidates efficiently and effectively.

For example, the following transformer from \Cref{fig:mlir-urem} is a valid initial transformer (we avoid using variable indices for readability and instead use actual variable names):
\begin{lstlisting}[language=MLIR, numbers=none]
func.func @f1(%L : KnownBits, %R : KnownBits) -> KnownBits {
    %1 = countLeadingZero(%L.zero) 
    %knownZero = setHighBits(0, %1)
    return makeKnownBits(%knownZero, 0)
}
\end{lstlisting}

That same transformer could be mutated into the following one by replacing the argument of \texttt{countLeadingZero} with a different input:
\begin{lstlisting}[language=MLIR, numbers=none]
func.func @f1(%L : KnownBits, %R : KnownBits) -> KnownBits {
    %1 = countLeadingZero(%R.zero)  // mutated
    %knownZero = setHighBits(0, %1)
    return makeKnownBits(%knownZero, 0)
}
\end{lstlisting}

\subsubsection{Randomly Adding Conditions to Unsound but Precise Transformers}
\label{sec:abduction}
Throughout a typical random search in which \Cref{alg:mcmc-best-transformer} uses the mutation strategy from \Cref{sec:randomly-mutating}, the algorithm discovers transformers that significantly improve precision but are not sound. 
Such transformers potentially contain valuable information, but cannot be used directly.
To address this issue, we use our MCMC approach to implement abductive synthesis, which constructs guards that identify subsets of the input space.
Given an unsound transformer $\tf$, we aim to synthesize a guard $c$ (written in the DSL $\lang$) and thereby construct a sound transformer $\condtf(a) := \fcpair{c}{\tf(a)}$.
We denote a conditional transformer as a term $\fcpair{c}{\tf}$.

\begin{algorithm}[tbp]
\footnotesize
{\it
\caption{\texttt{initialize} and \texttt{mutate} for condition abduction}
\label{alg:mcmc-abd}
\DontPrintSemicolon

\SetKwProg{Fn}{fun}{:}{}
\SetKw{None}{None}

\Fn{$\texttt{initializeAbd}()$ \label{alg:mcmc-abd-init}}{
    $\preccand\gets$ history of precise but unsound transformers encountered so far
    
    $f_u \gets$ random element from $\preccand$ \tcp*{select unsound but precise transformer} 
    
    $c \gets \texttt{initializeNormal}()$  \tcp*{sample a random condition expressible in $\lang$ using \Cref{alg:mcmc-normal}}
    
    \Return $\fcpair{c}{f_u}$ 
} 
\Fn{$\texttt{mutateAbd}(\fcpair{c}{f_u})$ \label{alg:mcmc-abd-mutate}}{
    $c' \gets \texttt{mutateNormal}(c)$  \tcp*{randomly mutate $c$ using DSL operators in $\lang$ using \Cref{alg:mcmc-normal}}
    \Return $\fcpair{c'}{f_u}$ 
}
}
\end{algorithm}

The \texttt{initializeAbd} and \texttt{mutateAbd} operations in \Cref{alg:mcmc-abd} implement this abductive synthesis by reusing the initialization and mutation operations in \Cref{alg:mcmc-normal} to explore the space of possible conditions $c$ starting from an unsound transformer.
The remaining structure of \Cref{alg:mcmc-best-transformer} is untouched as the cost function and overall structure of the algorithm are exactly the same regardless of what type of transformer we decide to synthesize.
For example, the following transformer \texttt{f\_2} from~\Cref{fig:mlir-urem} is a valid initial transformer for~\Cref{alg:mcmc-abd}, assuming again proper renaming of variables and that the last variable is being returned:

\begin{lstlisting}[language=MLIR, numbers=none]
func.func @f2_cond(%L: KnownBits, %R: KnownBits) -> bool {
    %Lmax = negate(%L.zero)
    %Rmin = %R.one
    %cond = unsignedLessThan(%Lmax, %Rmin)
    return %cond
}
func.func @f2_body(%L : KnownBits, %R : KnownBits) -> KnownBits { return %L }
func.func @f2(%L : KnownBits, %R : KnownBits) -> KnownBits {
    return ite(@f2_cond(%L, %R), @f2_body(%L, %R), %top)
}
\end{lstlisting}

That same transformer could be mutated by modifying \texttt{f2\_cond}, e.g., by changing the second argument of \texttt{unsignedLessThan}, but the body \texttt{f2\_body} would remain unchanged:
\begin{lstlisting}[language=MLIR, numbers=none]
func.func @f2_cond(%L : KnownBits, %R : KnownBits) -> KnownBits {
    %Lmax = negate(%L.zero)
    %Rmin = %R.one
    %cond = unsignedLessThan(%Lmax, %R.zero)  // mutated
    return %cond
}
\end{lstlisting}

%% file: 5implementation.tex
\section{Implementation}
\label{sec:implementation}


We have implemented \name as a modular and extensible synthesis framework that supports several \changed{finite} non-relational integer abstract domains and concrete operations.
\Cref{f:implementation} presents a high-level overview of the system architecture.
This section first describes how we instantiate the core synthesis algorithm to support different domains and instruction semantics (\Cref{s:impl:basic}), and then outlines key engineering optimizations that enable efficient large-scale synthesis (\Cref{s:impl:speed}).
The implementation is publicly available~\cite{xdsl-smt-artifact}.


\subsection{The Basic Ingredients}
\label{s:impl:basic}

Synthesis of an approximation of the ideal transformer $\besttf$ takes place through a two-loop process:
\rone A slow \textit{outer loop} performs bounded model checking (via z3) to validate soundness, and
\rtwo A fast \textit{inner loop} generates new candidates and \changed{aggressively tests their soundness and precision by sampling the space of abstract values.}
In practice, we run 1,000 inner-loop iterations per outer-loop iteration. After each outer iteration, the framework discards provably unsound candidates, \changed{adds sound candidates to the solution set $\soundset$,} and updates weights in the probability distribution used to sample MCMC mutations.

While the ultimate goal is to produce a set of sound abstract transformers $\soundset$, our implementation also maintains a set of unsound-but-precise functions, $\precset$, as well.
The purpose of $\precset$, as outlined in~\Cref{sec:abduction}, is to discover sound functions in a bottom-up way \changed{from a precise function $\tf \in \precset$ by synthesizing a condition $c$ that narrows the input space to a subset on which $\tf$ is sound.}
Concretely, the goal is to find $c$ such that $\condtf(x) := \fcpair{c(x)}{\tf(x)}$ is sound.
%

Accordingly, the system performs two interleaved forms of MCMC-guided synthesis. In each round of the inner loop:
\begin{itemize}
\item \changed{When discovering new sound transformers starting from a transformer $\tf$, synthesis mutates $\tf$ itself.}
\item \changed{When discovering guard conditions for existing transformers $\tf \in \precset$, synthesis keeps $\tf$ fixed and mutates $c$ in $\condtf(x) := \fcpair{c(x)}{\tf(x)}$.}
\end{itemize}

When the process converges or time runs out, \changed{we return the meet of all sound candidates $\soundsetmeet$.}
The precise set $\precset$ is discarded.

\subsubsection{Input Language}
\label{s:impl:input}

Our target DSL $\lang$ for synthesis is the MLIR dialect \texttt{xdsl-smt}~\cite{xdsl-smt}.
This language has several important characteristics:
(1) it includes basic numeric operations ($+$, $/$, \texttt{ite}), making it suitable to express a variety of transformers;
(2) it uses SSA form, which narrows the scope of possible mutations;
(3) it implements SMT-Lib~\cite{barrett2010smt} and has a direct mapping to z3; and
(4) it has a straightforward mapping to \cpp{}, enabling fast testing of transformers in the inner loop.

To instantiate our framework, users must provide the semantics of the concrete operation for which a transformer is to be synthesized, and the definition of the abstract domain over which the transformer operates.
\name targets one concrete operation at a time in its outer loop (hence the multi-arrow in~\Cref{f:implementation}).
Multiple loops can of course run in parallel.

For each concrete operation, the user provides both a z3 implementation (for \changed{verification in the outer loop}) and a corresponding \cpp{} implementation (for fast evaluation \changed{in the inner loop}).
Users \benchanged{must ensure} these are in agreement.
Operations may additionally specify input constraints \changed{prohibiting abstract values}---e.g., integer division excludes zero denominators.
For each abstract domain, the user supplies a top element, a meet operation, a concretization function (technically, only a membership test: the function is given a concrete value $v$ and an abstract value $A$, and must determine whether $v \in \gamma(A)$), a size function to guide precision, and an optional well-formedness constraint to rule out syntactically valid but semantically invalid abstract values (e.g., for \knownbits{}, the two bitvectors representing known-zero and known-one must not overlap).
Information about concrete operations can be reused across abstract domains.
Information about abstract domains can potentially be reused across different sets of concrete operations.

Turning back to our toy example from \Cref{se:problem-definition}, the following elements suffice to instantiate \name for \texttt{max} on the interval domain.
Each piece must be defined in MLIR:
\begin{minipage}[t]{0.49\columnwidth}
\begin{itemize}
\item Concrete op: $f(x, y) = \max(x, y)$
\item Input constraint: None
\item Top element: $[0, \texttt{UINT\_MAX}]$
\item Meet: $a \sqcap b = [\max(a.l, b.l), \min(a.r, b.r)]$
\end{itemize}\end{minipage}~\begin{minipage}[t]{0.49\columnwidth}
\begin{itemize}
\item Concretization: $\gamma([a.l, a.r]) = \{a.l, \cdots, a.r\}$
\item Size function: $\size{a} = \lfloor\log_2(\abs{a.l - a.r})\rfloor$
\item Well-formedness constraint: $\forall~a.\,a.l \leq a.r$
\end{itemize}\end{minipage}%


\subsubsection{Initializing MCMC in the Inner Loop}
\label{s:impl:inner}

Each iteration of the inner MCMC synthesis loop begins by generating a fresh pool of candidate transformer functions and candidate guard conditions.
Candidate functions are initialized with approximately 30 randomly generated instructions, drawn from the DSL described in \Cref{s:impl:input}. These functions are intended to represent full transformer logic. Candidate conditions, which are used for condition abduction~(\Cref{sec:abduction}), are initialized with around 6 instructions, reflecting their role as lightweight guards.
\changed{In both cases}, mutation attempts to morph \changed{these placeholder candidates into codes that outperform} the current members of $\soundset$.
We chose the constants 30 and 6 empirically; similar values work as well.

The scoring function used to evaluate candidates balances two objectives: soundness and precision. Two tunable parameters, $\kappa$ and $\lambda$~(from~\Cref{alg:mcmc-best-transformer}), govern this tradeoff. During the early stages of synthesis, we prioritize exploratory behavior by assigning a higher weight to precision for functions ($\lambda$) and to soundness for conditions ($\kappa$).
These weights encourage exploration; later on, the verifier prunes unsound functions and the scoring function de-emphasizes overly precise conditions.
The relative value of these parameters is important; the exact value is not so much.

\subsubsection{Mutation}
\label{s:impl:mutation}

Mutation is the simplest component of the implementation. It operates directly on MLIR programs expressed in the \texttt{xdsl-smt} dialect. Each candidate is a sequence of SSA assignments of the form $x_i = op(x_a, \ldots)$ where $op$ is an \texttt{xdsl-smt} operator, $i$ is the current line number, and $a$ is some index preceding $i$.
There are two possible mutations: replace one variable reference $x_a$ with another, or replace the operator with another $op'$ (with randomly selected arguments).
Our DSL consists of 29 operators in total, each of which corresponds to a function in \apint library.

Variables are always selected uniformly at random.
Each index in the range $[0, i)$ has equal probability.
Operators are initially chosen uniformly at random, but their weights get adjusted after each update to the sound set $\soundset$.
The operators that are most common across the sound functions have the highest chance of selection. 
Concretely, $\emph{weight}(op) := 1 + \emph{frequency}(op, \soundset)$, where $\emph{frequency}(op, F)$ is the total number of $op$ appearances in all transformers in $F$. 

Weights never decrease below a positive minimum, ensuring all operators maintain a nonzero probability of selection and preserving exploration diversity throughout synthesis.

\begin{figure}[t]\centering

  \begin{tikzpicture}
  \node (mlir) [align=center,fill=darkcolor] {{}\\ MLIR \\ xdsl-smt~\cite{xdsl-smt} \\};
  \node (rbracket) [above=0pt of mlir.north,anchor=south] {$\underbrace{\mbox{\begin{tabular}{c}Abstract Domains,\\ Concrete Ops\end{tabular}}}$};
  \node (synth) [right=0.7cm of mlir.east,anchor=west,align=center,fill=litecolor] {MCMC \\ Synthesis};
  \node (c0) [right=1.5cm of synth.east, anchor=west] {};
  \node (ctop) [above=0pt of c0.north, anchor=south, rectangle, draw] {\begin{tabular}{c}Candidate \\ Functions\end{tabular}};
  \node (cbot) [below=0pt of c0.south, anchor=north, rectangle, draw] {\begin{tabular}{c}Candidate \\ Conditions\end{tabular}};
  \node (verif) [right=1.5cm of c0.east, anchor=west, align=center, fill=darkcolor] {Verifier \\ z3};
  \node (d0) [right=0.5cm of verif.east, anchor=west] {~};
  \node (dtop) [above=-4pt of d0.north, anchor=south] {$\soundset$};
  \node (dbot) [below=-4pt of d0.south, anchor=north] {$\precset$};
  \node (halt) [right=0.4cm of d0.east, anchor=west, diamond, draw] {Halt?};
  \node (end) [right=0.2cm of halt.east, anchor=west] {$\meettf^s$};

  \node (mlir0) [above=2pt of mlir.east] {};
  \node (mlir1) [above=-2pt of mlir.east] {};
  \node (mlir2) [above=-6pt of mlir.east] {};
  \node (synth0) [above=2pt of synth.west] {};
  \node (synth1) [above=-2pt of synth.west] {};
  \node (synth2) [above=-6pt of synth.west] {};
  \draw [->] (mlir0) -- (synth0);
  \draw [->] (mlir1) -- (synth1);
  \draw [->] (mlir2) -- (synth2);
  \draw [->] (synth.east) -- (ctop.west);
  \draw [->] (synth.east) -- (cbot.west);
  \draw [->] (ctop.east) -- (verif.west);
  \draw [->] (cbot.east) -- (verif.west);
  \draw [->] (verif.east) -- (halt.west);
  \draw [->] (halt.east) -- (end.west);
  \node (feedback0) [above=0.7cm of halt] {};
  \node (feedback1) [above=1cm of synth] {};
  \draw [-] (halt.north) -- (feedback0.center);
  \draw [-] (feedback0.center) -- (feedback1.center);
  \draw [->] (feedback1.center) -- (synth.north);

  \node (boxnw) [below=0.9cm of synth.south west, xshift=-1cm] {};
  \node (boxne) [right=3.7cm of boxnw] {};
  \node (boxsw) [below=1cm of boxnw.center] {};
  \node (boxse) [below=1cm of boxne.center] {};
  \node (muta) [below=0.4cm of boxnw.center, anchor=west, xshift=0.2cm,fill=litecolor] {Mutate};
  \node (eval) [right=0.6cm of muta.east, anchor=west,fill=litecolor] {Evaluate};

  \draw [-,dashed,ultra thick,litecolor] (synth.south west) -- (boxnw.center);
  \draw [-,dashed,ultra thick,litecolor] (synth.south east) -- (boxne.center);
  \draw [-,ultra thick,litecolor] (boxne.center) -- (boxnw.center);
  \draw [-,ultra thick,litecolor] (boxne.center) -- (boxse.center);
  \draw [-,ultra thick,litecolor] (boxsw.center) -- (boxnw.center);
  \draw [-,ultra thick,litecolor] (boxsw.center) -- (boxse.center);
  \draw [->] (muta.east) -- (eval.west);
  \node (f3) [below=0.2cm of eval] {};
  \node (distfunc) [right=1pt of f3.north, anchor=west] {$\norm{\cdot}$};
  \node (f4) [below=0.2cm of muta] {};
  \draw [-] (eval.south) -- (f3.center);
  \draw [-] (f3.center) -- (f4.center);
  \draw [->] (f4.center) -- (muta.south);

  \end{tikzpicture}
\caption{Implementation overview. The outer loop (above) checks the soundness of promising candidates. The inner loop (below) generates new candidates via mutation and evaluates candidates' soundness and precision.}
\label{f:implementation}
\Description{Flowchart overview of \name{}. Given abstract domains and concrete ops in MLIR, we enter an inner loop that tests candidates. Promising candidates are fed into z3 for verification. We repeat the process and eventually output the meet of the best transformers.}
\end{figure}
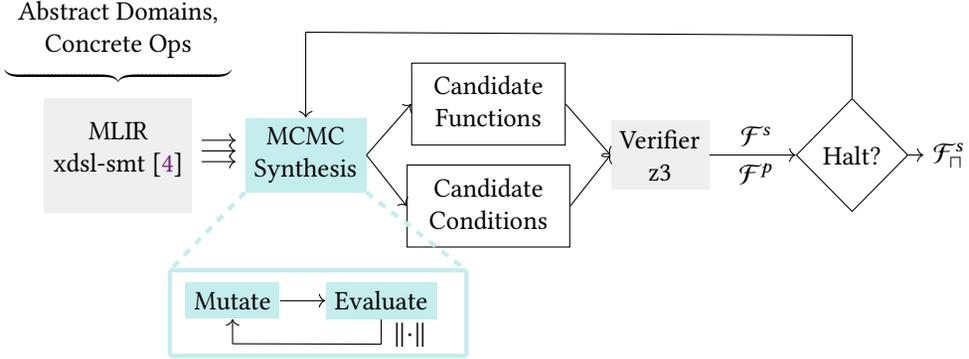

\subsubsection{Evaluation Engine}
\label{s:impl:evaluationengine}

After mutation generates candidate functions and conditions, a fast evaluation engine implemented in \cpp{} tests their soundness and precision by executing the semantics of the candidates. This engine implements the core procedure outlined in~\Cref{alg:mcmc-best-transformer}. We choose \cpp{} for its \benchanged{fast performance and good} integration with LLVM infrastructure---most notably, with LLVM's type \apint for arbitrary-precision integers.

Since our transformers operate over integer domains, the evaluation engine is tailored to efficiently test integer values in several ways:
\begin{itemize}
\item \textbf{Bitvector representation:}  The bitvectors that we use for static analysis are not faithfully represented by standard \cpp{} datatypes such as \texttt{uint64\_t}.
We use LLVM's \apint numbers instead.
Crucially, an \apint can be instantiated at any bitwidth and comes with numerous helper functions, e.g., for overflow checking.
\item  \textbf{Test generation by bitwidth:} Integers can be enumerated in a straightforward way. This property guides test generation:
\begin{itemize}
      \item For small bitwidths (range $[1,4]$), evaluation tests \changed{{all possible abstract values} and concrete inputs.} Candidate functions are thus guaranteed to be sound on small bitwidths.
      \changed{\emph{Examples:}~For \knownbits{}, there are 120 small abstract values.
      For \constantrange{}, there are 185 abstract values.
      To test a transformer $f^\sharp$, we compute $a = f^\sharp(a_0, a_1)$ for all pairs of abstract values $a_0, a_1$, and, unless $a$ is bottom,
      we check that $f(n_0, n_1) \in a$ for all pairs of concrete numbers.}
      
      \item For mid-size bitwidths ($[5,8]$), evaluation samples abstract values $a$ and exhaustively tests the concretization of $a$. Candidate functions are sound over sampled tests, but may be unsound over the entire abstract domain; z3 verification is there to catch mistakes.
      
      \item For large bitwidths ($[9,64]$), evaluation samples abstract values $a$ and samples elements $c \in \gamma(a)$. Due to this incomplete concretization, candidate functions could be unsound even on sampled abstract values at this stage.
    \end{itemize}
\end{itemize}

The extraction of an MLIR candidate to \cpp{} requires some care beyond the use of \apint values.
Operations that can lead to undefined behavior in \cpp{}, such as division, must be guarded with a check for invalid input and must return a default value that matches the behavior of the z3 verifier.
For division, invalid input is a zero denominator and the z3 default is to return zero.
Our goal here is simply to keep the evaluation engine and the verifier in sync.
If static analysis discovers that a program definitely divides by zero, other passes in LLVM will take appropriate action.


\subsubsection{Size Functions}
\label{s:impl:costfunction}

Each abstract domain has a size function $\size{\cdot}$ to quantify precision:
\begin{itemize}
\item  For the \knownbits domain, the size measures the number of unknown bits. 
For example, the \knownbits abstract value \texttt{01?1?00?} has size $3$.

\item  For the \constantrange domain, the size is computed as the $\lfloor \log_2 \rfloor$ of the absolute difference between the lower and upper bounds of the abstract intervals. 
For example a \constantrange abstract value of \texttt{[15, 45]} has a size of $\lfloor \log_2(30) \rfloor = 4$ 
\end{itemize}

\subsubsection{Integrating MCMC Results in the Outer Loop}

At each outer-loop iteration, the inner MCMC loop generates a collection of candidate transformers and guard conditions, ranked by their size. These candidates are then filtered and integrated as follows:

\begin{itemize}
\item Sound candidates that are not subsumed by any existing member of the solution set $\soundset$ are added to it. Subsumption is determined using the abstract domain’s ordering (e.g., inclusion).
\item Precise but unsound candidates are selectively retained in $\precset$, the pool used for condition abduction (\Cref{sec:abduction}). To control memory and evaluation cost, $\precset$ is capped at a fixed size (15 in our implementation), and retains only the top-scoring elements.
\end{itemize}

This strategy ensures that $\soundset$ grows monotonically with respect to both soundness and variety, while $\precset$ remains a focused source of potentially useful components for future synthesis.

\subsubsection{Verifier}
\label{s:impl:verifier}

Soundness is established via z3. For each candidate produced by the inner loop, we check it across all possible inputs for bitwidths ranging from 1 to 64. This validation ensures that the candidate respects the transformer specification over the full concrete input space.
\changed{If z3 times out, we simply assume the transformer is unsound and discard it.}
To perform this check, we use the lowering to z3 that is provided by the \texttt{xdsl-smt} dialect.
We use this lowering directly, though we did fix several bugs in it along the way~\cite{xdsl-smt-pull73,xdsl-smt-pull74}.
Candidates that pass verification are admitted to $\soundset$ and considered sound in subsequent synthesis and scoring phases.

\subsection{Gotta Go Fast}
\label{s:impl:speed}

\name can find promising transformers only because it can quickly explore a huge number of candidates.
Most candidates are garbage.
Sifting through the trash in a reasonable time is possible thanks to details of its engineering, described below.

\paragraph{Parallel MCMC}
\label{s:impl:parallel}

The MCMC synthesis in our inner loop is a slow, step-by-step process.
To accelerate converge and improve the likelihood of high coverage, we run this inner loop in parallel (currently, 100 times).
Operationally, each time we enter the inner loop with a set of candidate functions, we spawn several copies of MCMC that independently mutate and evaluate transformers.
At the end of the inner loop, the highest-scoring candidates from each process are selected as candidates for the outer loop to verify.
MCMC splits its focus between sound candidates ($\soundset$) and precise candidates ($\precset$) on a 70/30 basis:
\pct{70} of the parallel searches explore sound candidates,
and \pct{30} of the searches perform condition abduction.

\paragraph{LLVM JIT}
\label{s:impl:jit}


Our evaluation engine uses the LLVM JIT, \texttt{LLJIT}~\cite{lljit}, to quickly compile and test candidate functions.
We create a pre-compiled binary that packages the abstract domain, the \apint API, our scoring function, and LLJIT,
and use it to lower candidates written in MLIR to executable code.
Compared to our initial approach of invoking \texttt{clang}, linking with LLVM, and compiling 
the candidates, JIT-ing improves performance by roughly two orders of magnitude.
Testing a batch of roughly 100 candidates dropped from seconds to milliseconds with LLJIT.

\subsection{Generalizability of the Approach}

\changed{We have focused on \knownbits and \constantrange domains because these are the main abstract domains employed by LLVM.
There are no fundamental limitations to adapting \name to other abstract domains:}

\begin{itemize}
\item \emph{Non-relational abstract domains:}
\changed{\name can be directly applied to non-relational domains where each abstract value consists of a finite number of integers, as in \knownbits or non-wrapping interval domains.
Users must provide the components described above~(\Cref{s:impl:input}) and a soundness verifier.
Reuse is made simple because the template of mutated programs begins by deconstructing each input abstract value into integers and ends by reconstructing the output abstract value from several integers, and only the operations in the middle are mutated.
To support non-relational abstract domains that are not composed of integers, significant engineering is needed to define suitable templates and mutation strategies.}

\item \emph{Relational abstract domains:}
\changed{We conjecture that \name can synthesize transformers for relational domains, such as octagons and polyhedra.
We have deferred investigation because we suspect these domains will lead to prohibitively high compilation overhead.
In contrast to a non-relational domain, where program variables can be tracked independently,
a relational domain requires all variables to be tracked simultaneously.
A transformer must operate over entire program states rather than individual variables, which complicates reuse across operators.}

\item \emph{Non-Galois abstract domains:} \changed{
Our abstract domains are characterized by a Galois connection.
This means there is a unique most precise abstraction for each concrete set of integers,
a fact that our inner MCMC loop leverages to speed up evaluation.
To adapt \name to a non-Galois-connection-based abstract domain, the framework
would need a method for choosing among multiple sound abstract values.
Take the wrapped interval domain as an example~\cite{Gange2015Interval}.
There is no most precise abstraction
for integers in a circle; any number can be the left endpoint.
}
\end{itemize}

\changed{
\name targets MLIR, but there is nothing fundamental about our choice to use it.
MLIR is convenient because it has an SMT dialect~\cite{mlir-pldi25},
a dataflow analysis framework with pluggable transfer functions,
and a clear relation with LLVM to facilitate comparisons.
}

%% file: 7evaluation.tex
\section{Evaluation}
\label{se:evaluation}

\changed{Our evaluation consists of several parts:
a direct comparison between LLVM's \knownbits and \constantrange
transformers and ours~(\Cref{sec:comparison-to-llvm});
an evaluation of our known bits transfer functions vs.\ LLVM's, in
terms of precision and compilation time, when compiling an appropriate
subset of the SPEC 2017 benchmark suite~(\Cref{sec:end-to-end});
evidence that synthesis facilitates exploration on composite transfer
functions and a reduced product domain~(\Cref{sec:opportunities});
ablation studies of DSL choice~(\Cref{se:dsl}) and condition
abduction~(\Cref{se:evaluation:abduction});
and, a detailed comparison to
Amurth~\cite{DBLP:journals/pacmpl/KalitaMDRR22}~(\Cref{se:evaluation:amurth}).
Note that, although we compare against LLVM to show that our work is practical
beating LLVM at its own game is not a goal in and of itself.
\name is a complement to handcrafted analyses.
}

\input{tables/kb-table}

\input{tables/cr-table}





\subsection{Comparison to LLVM's Abstract Transformers}
\label{sec:comparison-to-llvm}

LLVM provides abstract transformers for a handful of abstract domains,
notably \texttt{KnownBits} and \texttt{ConstantRange}.
These transformers are handwritten and have been fine-tuned over the
years.  Currently, LLVM supports 44 of the \xuanyuchanged{47} binary operators for
\texttt{KnownBits} and \xuanyuchanged{37} of \xuanyuchanged{47} operators for \texttt{ConstantRange}.
To evaluate the effectiveness of our synthesis approach, we use our
tool to automatically generate abstract transformers for operators in
both domains.
%

Our evaluation reports data for \xuanyuchanged{40} of the \xuanyuchanged{47} operators.
The seven omitted operators---MulNsw, UMulSat, MulNuw, SAddSat,
MulNswNuw, SMulSat, and SSubSat---require overflow checks that we have
not yet supported in our dalect.

\subsubsection{Evaluation Setting}

We evaluate the precision of \name synthesized transformers, LLVM built-in transformers, and the meet of both.
We measure precision for 8-bit and 64-bit bitvectors on a set of randomly sampled inputs, denoted as $\adom_{\text{test}}$.
All concrete operators in our benchmarks are binary, so each input consists of a pair of abstract values $(a_1, a_2)$.
Generally speaking, the generation of test inputs for evaluation is the same as the test generation described in \Cref{s:impl:evaluationengine}, but using different random seeds.
%

In the 8-bit setting, the concrete domain is small (i.e., $2^8$), making it feasible to compute for each sampled input the theoretically best result $\besttf(a_1, a_2)$ by exhaustively enumerating all concrete value pairs $(c_1, c_2)$ such that $c_1 \in \gamma(a_1)$ and $c_2 \in \gamma(a_2)$.
We therefore report the percentage of inputs for which the evaluated transformer produces exactly the abstract output $\besttf(a_1, a_2)$.
In the 64-bit setting, exhaustive enumeration is no longer feasible as
each abstract value could concretize to as many as $2^{64}$ concrete
values.
Therefore, we report the norm over the sampled test inputs $\normp{\tf}{\adom_{\text{test}}}$.

Specifically, for \knownbits and \constantrange domain, we both sample 1,000 (and 10,000) pairs of abstract values for the 8-bit (and 64-bit) evaluations. For each abstract pair, we enumerate all concrete value pairs in the 8-bit setting, and sample 10,000 concrete value pairs in the 64-bit setting.

\mypar{Constraints on Concrete Operators}
As mentioned in \Cref{s:impl:input}, some concrete operators impose constraints on their operands.
For example, shift operators require the second operand to lie within the range $[0, \text{bitwidth}]$.
In such cases, a randomly sampled abstract input pair may not contain any concrete value pair that satisfies the constraint.
%
%
For 8-bit inputs, we apply rejection sampling to ensure only non-empty inputs---those that contain at least one valid concrete value pair---are included.
However, for 64-bit inputs, non-empty inputs could be rare, and rejection sampling becomes inefficient. Therefore, we simply skip empty inputs and exclude them from evaluation at 64 bits.
As a result, the \textbf{Tests} count in some rows of the right columns in \Cref{tab:kb-table,tab:cr-table} is less than 10,000.
In extreme cases where all 10,000 sampled inputs are empty (due to how few valid inputs there might be for a certain operator), we omit the corresponding data.

\mypar{Comparing to LLVM's wrapped \constantrange}
While the LLVM \knownbits domain cleanly fits into our framework, the LLVM \constantrange domain presents some friction, as it consists of wrapped intervals\cite{Gange2015Interval}.
Specifically, each element in this LLVM domain is either $\bot$ (the empty set), $\top$ (the set of all $w$-bit integers), or represented by $[a, b)$, where $a,b$ are $w$-bit bitvectors such that $a \neq b$.
This domain is a sign-agnostic domain.
The concretization function for LLVM \constantrange domain is defined as follows, where $<_l$ is lexicographic ordering on bit-vectors:
\[
    \gamma([a, b)) =
    \begin{cases}
        \{a, \ldots, b-1\}                           & \text{if } a <_l b \\
        \{0^w, \ldots, b-1\} \cup \{a, \ldots, 1^w\} & \text{otherwise.}
    \end{cases}
\]

The LLVM \constantrange domain \changed{is a sound but non-Galois} abstract domain---i.e., it does not form a Galois connection with the concrete bit-vector domain.
Its abstraction function $\alpha$ can be understood as returning a most precise wrapped interval that covers the given set of bitvectors---i.e., no subinterval of it does. However, such a most precise interval is not guaranteed to be unique.

\name only supports \changed{Galois-connection-based} abstract domains where a best abstraction is unique. So to enable a comparison to this non-well-defined abstract domain, we synthesized transformers for the following two segment domains:
\begin{itemize}
    \item Unsigned Intervals ($\ucr$): Each element is either $\bot$ or from the set $\{[a, b] \mid 0 \le a \le b < 2^w\}$.
    \item Signed Intervals ($\scr$): Each element is either $\bot$ or from the set $\{[a, b] \mid 2^{w-1} \le a \le b < 2^{w-1}\}$.
\end{itemize}
Because our abstract domains and the not-well-defined abstract domain used by LLVM are technically incomparable, we need to make some compromises when analyzing their precision and mapping elements of one to the other.
For unsigned concrete operators (e.g. \texttt{umax}) and sign-agnostic ones (e.g. \texttt{add}), $\ucr$ transformers are used for comparison. If the LLVM transformer produces a wrapper interval that cannot be represented as an element $\ucr$, we do not use that input in our comparison.
For signed operators (e.g., \texttt{smax}), we instead use $\scr$ transformers and similarly skip cases where LLVM’s wrapped interval cannot be captured in $\scr$.
To summarize, we only compare the two domains on inputs that they can both represent.
These skipped inputs cause the \textbf{Tests} number in both 8-bit and 64-bit in \Cref{tab:cr-table} to be less than the desired number of sampled inputs.




\subsubsection{Evaluation Results}

\Cref{tab:kb-table} and \Cref{tab:cr-table} summarize \knownbits and \constantrange evaluation result respectively.
Experiments ran on a server with 2× Intel Xeon Gold 6230 CPUs (40 cores, 80 threads, 2.10GHz). Each benchmark took around 6 hours to finish.

For the \knownbits domain, LLVM provides manually written transformers for $37/40$ concrete operators.
          Among these, there are \benchanged{$12$} operators for which the LLVM transformers do not achieve the theoretically best result.
          For \benchanged{$3$} out of these \benchanged{$12$} cases, \name synthesizes more precise transformers, as measured by the norm function on 64-bit inputs.
          For another \benchanged{$5$} out of \benchanged{$11$} cases, the synthesized transformers are less precise than the LLVM ones but still uncover new heuristics that are not present in the LLVM implementations and therefore result in increased precision for the meet of the LLVM and the synthesized transformer.
          The remaining $28$ out of $40$ LLVM transformers are already the best, meaning there is no room for improvement via synthesis.
          That being said, \name is able to match the precision for $3$ of these hand-tuned transformers.
          \name also synthesizes transformers with non-trivial precision for the two operations for which LLVM does not include a transformer.

For the \constantrange domain, LLVM provides manually written transformers for $30/40$ concrete operators.
          \benchanged{\name{} is able to synthesize transformers for all $40$ operators.}
          There are \benchanged{$13$} operators where LLVM transformers do not achieve the theoretically best result.
          For \benchanged{$9$} out of these \benchanged{$13$} cases, \name synthesizes more precise transformers, as measured by the norm function on 64-bit inputs.
          For another \benchanged{$6$} out of \benchanged{$13$} cases, the synthesized transformers are less precise than the LLVM ones but still uncover new heuristics that are not present in the LLVM implementations (again shown by the increased precision obtained when taking the meet of the LLVM and synthesized transformer).
          The remaining \benchanged{$17$} out of \benchanged{$30$} LLVM transformers are already the best, meaning there is no room for improvement via synthesis.
          \name matches existing hand-tuned transformers for \benchanged{$11$} of these cases.

\subsection{End-to-End Precision and Performance}
\label{sec:end-to-end}

\changed{In this section we evaluate the effect of replacing LLVM's
  known bits transfer functions with our own.
We perform this comparison to show that our transfer functions are
reasonable ones, but our larger goal is not to beat LLVM at its own
game.
Rather, we aim to develop basic technologies that can be used to avoid
manual implementation of transfer functions in future compilers.}

\changed{Replacing LLVM's transfer functions with our own was not entirely a
clean software engineering job since LLVM's known bits analysis is
implemented in a style that intermixes its transfer functions with a
highly ad hoc dataflow framework that simply recurses along backwards
dataflow edges until a maximum depth is reached.
We left this framework in place, but called out to our own transfer
functions, instead of LLVM's, at appropriate points in the code.}

\changed{We took the nine C/C++, integer programs from the SPEC CPU 2017
benchmark suite and compiled them using an off-the-shelf Clang/LLVM
version 21, and then also our modified version.
This optimizing compilation (using Clang's \texttt{-O3 -march=native}
flags) was done on a Linux machine using an AMD 2990WX 32-core CPU,
with parallelism disabled in the build system.
During an optimizing compile, LLVM uses the known bits analysis
results many times, from many different passes.
To evaluate precision, we took the final, optimized LLVM IR resulting
from optimized compilation and invoked LLVM's known bits
analysis---with and without our synthesized transfer functions---on
every integer-typed SSA value.}

\input{tables/compilation-time-table}

\changed{Table~\ref{tab:compile-time} shows the results of this
  experiment.
Our transfer functions are neither as fast nor as precise as LLVM's,
but the difference is not huge.
We have not yet attempted to optimize our transfer functions for
runtime performance.
One aspect that could be improved is our generated C++, which is not
idiomatic: it creates and destroys more temporary objects than does
LLVM's hand-written code, and furthermore it performs some translation
between MLIR and LLVM data structures that could be elided.
Another potential avenue for improvement would be to make execution
time part of our fitness function during synthesis.}

\benchanged{In addition, we ran a smaller ad-hoc experiment
compiling openssl, ffmpeg, and cvc5 using the meet of our synthesized
\knownbits transformers and LLVM's.
The precision improvements are modest:
+2 discovered bits in openssl (above the baseline of 1.3M bits discovered by LLVM),
+14 bits in ffmpeg (baseline: 3.7M), and
+0 bits in cvc5 (baseline: 16M).}

\subsection{Opportunities Created by \name{}}
\label{sec:opportunities}

\changed{We believe that synthesis changes the basic economics of transfer
functions, in the sense that we can synthesize more of these than we
would want to write by hand.
%
}

\subsubsection{Precision via Specialization}

\changed{Production compilers usually provide a wide variety of intrinsic
functions that encapsulate higher-level operations such as popcount
(Hamming weight) and saturating arithmetic operations.
The advantage of intrinsics over open-coded versions of these
operations is that the composite versions can readily be lowered to
either dedicated machine instructions or optimized library code.
However, composite operations have a secondary benefit, which is that
a composite transfer function is typically more precise than what we
would get by composing the results of the transfer functions for the
elementary operations that make up the open-coded version of the
operation.
To create Table~\ref{tab:intrinsics}, which illustrates this point, we
synthesized transfer functions for eight LLVM intrinsics, and then we
also measured the precision attained by analyzing an open-coded
version of each operation---that is, by composing the most precise
transformers that we have been able to synthesize.
In all cases, the transfer function for the composed operation is
considerably more precise.}

\input{tables/intrinsics}

\subsubsection{Reduced Product}

\input{tables/rp-table}

A reduced product domain improves analysis precision by combining the strengths of multiple abstract domains. Each domain captures different aspects of program behavior---for example, \knownbits tracks known zero and one bits, while \constantrange tracks possible value ranges.
The reduced product coordinates these domains using a reduction operator ($\sigma$), which refines each domain's element based on information from the other~\changed{\cite{cc-popl-1979}}. This mutual refinement helps eliminate infeasible states and improve the overall precision of the analysis.
Since \name synthesizes transformers for all operations in both the \knownbits and \constantrange domains, we can automatically construct reduced product transformers for any concrete operation by manually providing a suitable reduction operator that relates the two domains.

We evaluate the difference in precision between using \knownbits and using the reduced product between \knownbits and \constantrange.
To have a uniform random sampling from this reduced product lattice, we simply sample uniformly from both \knownbits and \constantrange, where the reduction between the \knownbits and \constantrange abstract values is not bottom.

We compare only against our synthesized transformers, not LLVM’s, because our sampling occurs over the product lattice. This means the abstract elements we evaluate are inherently more precise than those representable by \knownbits alone, making a direct comparison with LLVM’s \knownbits transformers unfair.
Instead, we focus on identifying which concrete operations benefit from the added precision of the reduced product domain, and by how much, when using our synthesized transformers.

\Cref{tab:rp-table} presents our results.
The experiment finished quickly (90 seconds, on a reasonably fast
laptop) because it merely evaluates transformers rather than
synthesizing transformers.
We report only the transformers for which the reduced product yields a precision improvement over using \knownbits alone.
On the remaining operations, such as \texttt{Xor} (not shown in \Cref{tab:rp-table}), the reduced product offers no additional benefit, as \knownbits is already maximally precise. In contrast, operations such as \texttt{AddNswNuw} show substantial gains in precision when evaluated within the reduced product domain.
There were $22$ such concrete operations for which the reduced product was able to improve over the \knownbits domain, out of \benchanged{$40$} total operations.
For the concrete operations where there was an improvement, the reduced product improved by an average of $24.4\%$ points for the tests at $8$-bits wide, and improved the average size score by $196.2$, on the tests at $64$-bits wide.
By combining \knownbits with \constantrange in a reduced product, we are able to improve precision particularly for operators that are challenging for both our synthesized \knownbits transformers and LLVM’s. Overall, the reduced product delivers consistently strong performance across the full range of supported operations, thus showcasing one additional benefit enabled by \name.

\subsection{Impact of DSL Choice}
\label{se:dsl}

\changed{We evaluate how the choice of operations in DSL affects the performance of \name.
    The \textbf{Full} DSL used in \Cref{sec:comparison-to-llvm} includes 29 primitives, which can be grouped as follows:}
\changed{
    (1) \emph{Bitwise:} \texttt{and}, \texttt{or}, \texttt{xor}, \texttt{neg};
    (2) \emph{Add:} \texttt{add}, \texttt{sub}.
    (3) \emph{Max:} \texttt{umax}, \texttt{umin}, \texttt{smax}, \texttt{smin};
    (4) \emph{Mul:} \texttt{mul}, \texttt{udiv}, \texttt{sdiv}, \texttt{urem}, \texttt{srem};
    (5) \emph{Shift:} \texttt{shl}, \texttt{ashr}, \texttt{lshr};
    (6) \emph{BitSet:} \texttt{set\_high\_bits}, \texttt{set\_low\_bits}, \texttt{clear\_low\_bits}, \texttt{clear\_high\_bits}, \texttt{set\_sign\_bit}, \texttt{clear\_sign\_bit};
    (7) \emph{BitCount:} \texttt{count\_left\_one}, \texttt{count\_left\_zero}, \texttt{count\_right\_one}, \texttt{count\_right\_zero}.
    (8) \emph{ITE:} \texttt{if\_then\_else}.
}
\changed{
We conduct an ablation study on the \knownbits domain over two DSL subsets:
\basicops=\emph{Bitwise} $\cup$ \emph{Add}, contains only addition, subtraction, and bitwise operations.
\bitops=\emph{Bitwise} $\cup$ \emph{Add} $\cup$ \emph{Max} $\cup$ \emph{Shift} $\cup$ \emph{BitSet} $\cup$ \emph{BitCount} $\cup$ \emph{ITE}, contains all primitives except those in \emph{Mul}.
}

\changed{
    %
    The \textbf{Full}  language, \basicops, and \bitops yield the most precise transformers for 16, 5, and 14 operations, respectively.
    For the remaining 5/40 benchmarks, they all reach the same precision.
    \whenappendix{Detailed results are summarized in \Cref{tab:dsl-table} in \Cref{app:ablation:dsl}.}
}

\changed{The results are in a way expected: if certain operations are known to be irrelevant for a specific transformer, removing them from the DSL can improve the precision---e.g., \basicops is the smallest of the three languages that can produce an optimal transformer for \texttt{add} and \texttt{sub} and does well for such benchmarks.}
    %



\subsection{Impact of Abduction}
\label{se:evaluation:abduction}

\changed{We now evaluate the impact of condition abduction (\Cref{sec:abduction}) by running \name with and without condition abduction.
    For \knownbits, condition abduction improves precision by $6.44\%$ on average (geometric mean), with 19/40 benchmarks showing gains.
    For \constantrange, condition abduction improves precision by $2.3\%$ on average (geometric mean), with 16/37 benchmarks showing gains.
    \whenappendix{Detailed Results are shown in Tables~\ref{tab:kb-abduction} and~\ref{tab:cr-abduction} in \Cref{app:ablation:abduction}.}
}

\changed{A representative case highlighting the necessity of abduction is the \texttt{add} benchmark in the $\ucr$ domain.
Its best transformers produce intervals better than $\top$ only when one knows overflow cannot happen.
Using condition abduction, \name synthesizes the best transformer in 2 rounds: first, it discovers a transformer that only works for non-overflow cases, and then it identifies the overflow condition to generate a complete transformer.
When abduction is disabled, \name cannot find the best transformer within 5 rounds.}

\changed{To summarize, overall condition abduction generally improves precision, but because it uses part of the compute budget, it may in certain cases reduce precision.
Of course, one can also run \name with and without abduction and output the better result.}

\subsection{\changed{Comparison to \amurth}}
\label{se:evaluation:amurth}

\changed{
The only other work we are aware of that tackles the general problem of synthesizing abstract transformers is \amurth~\cite{DBLP:journals/pacmpl/KalitaMDRR22}.
\amurth takes the same inputs as \name---i.e., a DSL, a concrete semantics, and an abstract domain---and synthesizes a provably most precise (up to a given input bound) transformer expressible within the given DSL.
To do so, \amurth uses constraint solvers (specifically Sketch~\cite{DBLP:journals/sttt/Solar-Lezama13}) to iteratively synthesize transformers that increase precision on a finite set of abstract inputs until a provably optimal transformer is sound.}

\changed{
The evaluation by \citeauthor{DBLP:journals/pacmpl/KalitaMDRR22} includes transformers for both string operations and integer operations. Because \name currently does not support string operations, we only focus on the latter.
\amurth has successfully been used to synthesize transformers for 9 concrete operators for the unsigned and signed interval domains \citet[Table 4]{DBLP:journals/pacmpl/KalitaMDRR22}: \texttt{add}, \texttt{sub}, \texttt{mul}, \texttt{and}, \texttt{or}, \texttt{xor}, \texttt{shl}, \texttt{ashr}, and \texttt{lshr}.
We therefore focus our evaluation on these 2 domains and  9 operations.
}

\changed{
When provided with this DSL consisting of the set of 29 primitive instructions we used in \Cref{se:dsl}, \amurth could not synthesize any transformer or returned $\top$ within the time limit.
}

\changed{
While \amurth cannot synthesize abstract transformers when given a generic DSL, it can do so by providing ``hints'' to the synthesizer in the form of sketches---i.e., partial programs where only some parts are missing---and custom auxiliary functions.
For example, \amurth synthesizes transformers for bitwise operators (\texttt{and}, \texttt{or}, \texttt{xor}), when auxiliary functions such as \texttt{minOr}, \texttt{maxOr}, \texttt{minAnd}, \texttt{maxAnd} (which compute the lower/upper bound of the results of bitwise-or/and over 2 intervals) are provided.
With those auxiliary functions provided, they further provide program template (that describes a clever way to divide input intervals at 0) to synthesize for signed domains \citet[Figure 16]{DBLP:journals/pacmpl/KalitaMDRR22}. 
Providing hints and templates allows \amurth to synthesize most-precise transformers for very tricky transformers operations, but requires the users of \amurth to provide insights that are quite close to the actual solution.
Moreover, one has to provide \amurth with different hints and sketches (i.e., different DSLs) for different concrete operations, even when the underling abstract domain does not change.
}

\changed{To summarize, \amurth cannot solve the problem tackled in this paper---i.e., synthesize abstract transformers for many concrete operators using just \textit{one} given DSL.
However, \amurth is well-suited for synthesizing optimal transformers for tricky individual operations, as long as the user is willing to provide hints to the synthesizer in the form of program sketches and auxiliary functions.
}


%% file: tables/kb-table.tex
\begin{table}[p] 
\centering
\caption{\knownbits results for 8-bit and 64-bit integers.
The 8-bit results report the percentage of tests where the transformer matches the best transformer $\besttf$.
The 64-bit results report the norm of the transformer \changed{(normalized by the number of \textbf{Tests})} which reflects its imprecision.
\changed{Higher ($\uparrow$) is better for \textbf{exact}, while lower ($\downarrow$) is better for \textbf{norm}.}
\#$\tf$ is the number of transformers in $\soundset$. \#$c$ counts those with conditions. \#inst is the total number of MLIR instructions used by those transformers.
\textbf{Tests} is the number of sampled test inputs. $\top$, synth,
and LLVM present results for three transformers: a trivial one that
always returns top, our synthesized transformer, and LLVM's
manually-implemented transformer. The meet column is for the
transformer {synth $\sqcap$ LLVM}.
\textbf{Bolded} numbers denote cases where the meet of the synthesized transformer and the LLVM one has higher precision than the LLVM one. If no LLVM implementation is available (N/A) we bold cases in which \changed{the synthesized transformer is better than $\top$}.
Some 64-bit rows contain a dash (-) because random sampling failed to produce any valid pairs of inputs. These operators have strict constraints on their inputs.
}
\label{tab:kb-table}
\footnotesize
\setlength{\tabcolsep}{3.5pt} 
\begin{tabular}{@{}|l
                rrr
                |r
                rrrr
                |r
                rrrr@{}}
\toprule
\multirow{2}{*}{\textbf{ConcreteOp}} & \multicolumn{3}{c}{} & \multirow{2}{*}{\textbf{Tests}} & \multicolumn{4}{c}{\textbf{8-bit exact (\%)} $\uparrow$} & \multirow{2}{*}{\textbf{Tests}} & \multicolumn{4}{c}{\textbf{64-bit precision (norm)} $\downarrow$} \\
\cmidrule(lr){2-4} \cmidrule(lr){6-9} \cmidrule(lr){11-14}
 & \#$\tf$ & \#$c$ & \#inst & & $\top$ & synth & llvm & meet & & $\top$ & synth & llvm & meet \\
\midrule
Abds & 10 & 3 & 189 & 1000 & 33.90 & 60.10 & 100.00 & 100.00 & 10000 & 0.059 & 0.050 & 0.000 & 0.000 \\
Abdu & 16 & 4 & 259 & 1000 & 33.10 & 59.40 & 100.00 & 100.00 & 10000 & 0.059 & 0.050 & 0.000 & 0.000 \\
Add & 17 & 2 & 306 & 1000 & 29.60 & 58.70 & 100.00 & 100.00 & 10000 & 0.140 & 0.082 & 0.000 & 0.000 \\
AddNsw & 13 & 2 & 205 & 1000 & 24.50 & 42.00 & 100.00 & 100.00 & 9674 & 0.147 & 0.136 & 0.000 & 0.000 \\
AddNswNuw & 14 & 3 & 220 & 1000 & 7.40 & 45.50 & 100.00 & 100.00 & 7479 & 0.160 & 0.136 & 0.000 & 0.000 \\
AddNuw & 17 & 4 & 291 & 1000 & 15.20 & 53.90 & 100.00 & 100.00 & 8305 & 0.152 & 0.103 & 0.000 & 0.000 \\
And & 1 & 0 & 14 & 1000 & 0.10 & 100.00 & 100.00 & 100.00 & 10000 & 0.625 & 0.000 & 0.000 & 0.000 \\
Ashr & 5 & 3 & 94 & 1000 & 31.10 & 65.50 & 85.70 & 85.70 & 0 & - & - & - & - \\
AshrExact & 7 & 2 & 123 & 1000 & 14.70 & 40.10 & 100.00 & 100.00 & 0 & - & - & - & - \\
AvgCeilS & 8 & 4 & 157 & 1000 & 31.80 & 38.80 & 100.00 & 100.00 & 10000 & 0.139 & 0.136 & 0.000 & 0.000 \\
AvgCeilU & 8 & 6 & 153 & 1000 & 31.60 & 38.60 & 100.00 & 100.00 & 10000 & 0.139 & 0.136 & 0.000 & 0.000 \\
AvgFloorS & 9 & 3 & 163 & 1000 & 32.40 & 39.30 & 100.00 & 100.00 & 10000 & 0.139 & 0.136 & 0.000 & 0.000 \\
AvgFloorU & 7 & 4 & 138 & 1000 & 32.20 & 37.70 & 100.00 & 100.00 & 10000 & 0.139 & 0.136 & 0.000 & 0.000 \\
Lshr & 4 & 1 & 80 & 1000 & 12.30 & 59.30 & 96.50 & 96.50 & 0 & - & - & - & - \\
LshrExact & 6 & 2 & 119 & 1000 & 12.80 & 31.40 & 100.00 & 100.00 & 0 & - & - & - & - \\
Mods & 12 & 3 & 235 & 1000 & 41.20 & 64.70 & 71.30 & \textbf{71.50} & 10000 & 0.090 & 0.078 & 0.076 & 0.076 \\
Modu & 12 & 4 & 207 & 1000 & 16.70 & 59.00 & 52.70 & \textbf{70.60} & 10000 & 0.149 & 0.030 & 0.132 & \textbf{0.027} \\
Mul & 10 & 6 & 202 & 1000 & 25.60 & 60.60 & 73.20 & \textbf{73.30} & 10000 & 0.025 & 0.010 & 0.006 & \textbf{0.006} \\
Or & 1 & 0 & 8 & 1000 & 0.00 & 100.00 & 100.00 & 100.00 & 10000 & 0.624 & 0.000 & 0.000 & 0.000 \\
Sdiv & 11 & 7 & 248 & 1000 & 64.10 & 64.10 & 83.40 & 83.40 & 10000 & 0.331 & 0.331 & 0.114 & 0.114 \\
SdivExact & 2 & 2 & 32 & 1000 & 19.30 & 19.30 & 37.40 & 37.40 & 0 & - & - & - & - \\
Shl & 4 & 1 & 69 & 1000 & 10.50 & 56.90 & 96.50 & 96.50 & 0 & - & - & - & - \\
ShlNsw & 7 & 1 & 115 & 1000 & 6.80 & 26.20 & 100.00 & 100.00 & 0 & - & - & - & - \\
ShlNswNuw & 7 & 2 & 115 & 1000 & 5.50 & 9.80 & 100.00 & 100.00 & 0 & - & - & - & - \\
ShlNuw & 7 & 3 & 139 & 1000 & 10.60 & 40.50 & 100.00 & 100.00 & 0 & - & - & - & - \\
Smax & 9 & 5 & 185 & 1000 & 6.50 & 63.80 & 100.00 & 100.00 & 10000 & 0.348 & 0.095 & 0.000 & 0.000 \\
Smin & 6 & 4 & 119 & 1000 & 6.00 & 72.80 & 100.00 & 100.00 & 10000 & 0.349 & 0.064 & 0.000 & 0.000 \\
SshlSat & 7 & 1 & 124 & 1000 & 37.60 & 72.40 & N/A & \textbf{72.40} & 10000 & 0.624 & 0.109 & N/A & \textbf{0.109} \\
Sub & 12 & 2 & 204 & 1000 & 28.50 & 60.60 & 100.00 & 100.00 & 10000 & 0.140 & 0.088 & 0.000 & 0.000 \\
SubNsw & 16 & 5 & 287 & 1000 & 22.20 & 47.80 & 100.00 & 100.00 & 9673 & 0.146 & 0.103 & 0.000 & 0.000 \\
SubNswNuw & 14 & 6 & 292 & 1000 & 7.60 & 31.80 & 100.00 & 100.00 & 7515 & 0.160 & 0.139 & 0.000 & 0.000 \\
SubNuw & 18 & 4 & 323 & 1000 & 16.40 & 47.10 & 100.00 & 100.00 & 8215 & 0.152 & 0.099 & 0.000 & 0.000 \\
UaddSat & 14 & 5 & 305 & 1000 & 18.30 & 61.80 & 100.00 & 100.00 & 10000 & 0.253 & 0.069 & 0.000 & 0.000 \\
Udiv & 10 & 5 & 191 & 1000 & 2.50 & 80.80 & 89.80 & \textbf{90.90} & 10000 & 0.960 & 0.004 & 0.001 & 0.001 \\
UdivExact & 3 & 2 & 53 & 1000 & 2.80 & 15.30 & 33.90 & \textbf{42.20} & 0 & - & - & - & - \\
Umax & 9 & 6 & 199 & 1000 & 6.40 & 90.60 & 100.00 & 100.00 & 10000 & 0.351 & 0.002 & 0.000 & 0.000 \\
Umin & 6 & 1 & 105 & 1000 & 6.40 & 92.90 & 100.00 & 100.00 & 10000 & 0.351 & 0.001 & 0.000 & 0.000 \\
UshlSat & 2 & 1 & 41 & 1000 & 3.60 & 96.60 & N/A & \textbf{96.60} & 10000 & 1.000 & 0.000 & N/A & {0.000} \\ 
UsubSat & 10 & 6 & 232 & 1000 & 19.00 & 52.10 & 100.00 & 100.00 & 10000 & 0.254 & 0.071 & 0.000 & 0.000 \\
Xor & 3 & 0 & 39 & 1000 & 2.30 & 100.00 & 100.00 & 100.00 & 10000 & 0.390 & 0.000 & 0.000 & 0.000 \\
\bottomrule
\end{tabular}
\end{table}

%% file: tables/cr-table.tex
\begin{table}[t]
\centering
\caption{\constantrange results for 8-bit and 64-bit integers.
\benchanged{See the~\Cref{tab:kb-table} caption for descriptions of the columns, \textbf{exact},
and \textbf{norm}.}
\changed{Higher ($\uparrow$) is better for \textbf{exact}. Lower ($\downarrow$) is better for \textbf{norm}.}
Operators marked by a * use synthesized $\scr$ transformers, others use $\ucr$ transformers.
}

\label{tab:cr-table}
\footnotesize
\setlength{\tabcolsep}{3.5pt}
\begin{tabular}{@{}|l
                rrr
                |r
                rrrr
                |r
                rrrr@{}}
\toprule
\multirow{2}{*}{\textbf{ConcreteOp}} & \multicolumn{3}{c}{} & \multirow{2}{*}{\textbf{Tests}} & \multicolumn{4}{c}{\textbf{8-bit exact (\%)} $\uparrow$} & \multirow{2}{*}{\textbf{Tests}} & \multicolumn{4}{c}{\textbf{64-bit precision (norm)} $\downarrow$} \\
\cmidrule(lr){2-4} \cmidrule(lr){6-9} \cmidrule(lr){11-14}
 & \#$\tf$ & \#$c$ & \#inst & & $\top$ & synth & llvm & meet & & $\top$ & synth & llvm & meet \\
\midrule
Abds* & 3 & 2 & 70 & 1000 & 59.80 & 59.80 & N/A & \textbf{59.80} & 10000 & 0.917 & 0.915 & N/A & \textbf{0.915} \\
Abdu & 20 & 6 & 344 & 1000 & 0.00 & 75.00 & N/A & \textbf{75.00} & 10000 & 0.990 & 0.908 & N/A & \textbf{0.908} \\
Add & 2 & 2 & 45 & 509 & 36.54 & 100.00 & 100.00 & 100.00 & 4991 & 0.949 & 0.887 & 0.887 & 0.887 \\
AddNsw* & 8 & 4 & 148 & 1000 & 7.10 & 100.00 & 100.00 & 100.00 & 9770 & 0.982 & 0.905 & 0.905 & 0.905 \\
AddNswNuw & 10 & 2 & 172 & 1000 & 0.00 & 70.70 & 84.80 & \textbf{88.60} & 8190 & 0.994 & 0.921 & 0.912 & \textbf{0.910} \\
AddNuw & 14 & 5 & 243 & 1000 & 0.00 & 94.00 & 100.00 & 100.00 & 8267 & 0.993 & 0.910 & 0.906 & 0.906 \\
And & 10 & 5 & 193 & 1000 & 0.00 & 82.30 & 83.10 & \textbf{83.30} & 10000 & 0.990 & 0.886 & 0.883 & 0.883 \\
Ashr* & 8 & 3 & 141 & 1000 & 0.00 & 98.00 & 98.50 & \textbf{99.20} & 0 & - & - & - & - \\
AshrExact* & 4 & 2 & 73 & 1000 & 0.00 & 81.70 & N/A & \textbf{81.70} & 0 & - & - & - & - \\
AvgCeilS* & 26 & 2 & 444 & 1000 & 0.00 & 0.30 & N/A & \textbf{0.30} & 10000 & 0.998 & 0.953 & N/A & \textbf{0.953} \\
AvgCeilU & 17 & 3 & 303 & 1000 & 0.00 & 1.80 & N/A & \textbf{1.80} & 10000 & 0.998 & 0.947 & N/A & \textbf{0.947} \\
AvgFloorS* & 20 & 1 & 332 & 1000 & 0.00 & 1.80 & N/A & \textbf{1.80} & 10000 & 0.998 & 0.955 & N/A & \textbf{0.955} \\
AvgFloorU & 17 & 3 & 276 & 1000 & 0.00 & 30.00 & N/A & \textbf{30.00} & 10000 & 0.998 & 0.932 & N/A & \textbf{0.932} \\
Lshr & 1 & 0 & 14 & 1000 & 0.00 & 100.00 & 100.00 & 100.00 & 0 & - & - & - & - \\
LshrExact & 4 & 1 & 78 & 1000 & 0.00 & 85.20 & N/A & \textbf{85.20} & 0 & - & - & - & - \\
Mods* & 6 & 2 & 118 & 1000 & 0.00 & 43.50 & 95.60 & 95.60 & 10000 & 0.997 & 0.909 & 0.875 & 0.875 \\
Modu & 8 & 3 & 164 & 1000 & 0.00 & 89.00 & 88.90 & \textbf{89.00} & 10000 & 0.989 & 0.877 & 0.877 & \textbf{0.877} \\
Mul & 17 & 6 & 371 & 998 & 90.78 & 90.78 & 90.88 & 90.88 & 10000 & 0.861 & 0.861 & 0.861 & 0.861 \\
Or & 14 & 8 & 275 & 1000 & 0.00 & 82.40 & 85.60 & \textbf{86.20} & 10000 & 0.990 & 0.889 & 0.882 & 0.882 \\
Sdiv* & 8 & 3 & 146 & 998 & 0.50 & 12.32 & 100.00 & 100.00 & 10000 & 1.000 & 0.997 & 0.501 & 0.501 \\
SdivExact* & 8 & 8 & 202 & 1000 & 0.70 & 0.70 & N/A & \textbf{0.70} & 0 & - & - & - & - \\
Shl & 8 & 4 & 140 & 1000 & 0.00 & 83.00 & 0.30 & \textbf{83.00} & 0 & - & - & - & - \\
ShlNsw* & 4 & 3 & 81 & 1000 & 0.70 & 18.00 & 99.20 & \textbf{99.40} & 0 & - & - & - & - \\
ShlNswNuw & 7 & 1 & 116 & 1000 & 0.00 & 88.30 & 49.80 & \textbf{100.00} & 0 & - & - & - & - \\
ShlNuw & 7 & 4 & 156 & 1000 & 0.00 & 76.40 & 99.60 & 99.60 & 0 & - & - & - & - \\
Smax* & 1 & 0 & 14 & 1000 & 0.00 & 100.00 & 100.00 & 100.00 & 10000 & 0.996 & 0.838 & 0.838 & 0.838 \\
Smin* & 1 & 0 & 13 & 1000 & 0.00 & 100.00 & 100.00 & 100.00 & 10000 & 0.996 & 0.839 & 0.839 & 0.839 \\
SshlSat* & 7 & 3 & 132 & 1000 & 50.90 & 75.50 & 100.00 & 100.00 & 10000 & 0.497 & 0.486 & 0.243 & 0.243 \\
Sub & 2 & 2 & 39 & 526 & 35.36 & 100.00 & 100.00 & 100.00 & 4997 & 0.948 & 0.887 & 0.887 & 0.887 \\
SubNsw* & 5 & 4 & 106 & 1000 & 7.00 & 100.00 & 100.00 & 100.00 & 9773 & 0.981 & 0.904 & 0.904 & 0.904 \\
SubNswNuw & 3 & 1 & 59 & 1000 & 0.00 & 74.50 & 84.00 & 84.00 & 8152 & 0.993 & 0.917 & 0.912 & 0.912 \\
SubNuw & 14 & 8 & 259 & 1000 & 0.00 & 94.50 & 100.00 & 100.00 & 8304 & 0.993 & 0.910 & 0.906 & 0.906 \\
UaddSat & 6 & 1 & 93 & 1000 & 0.00 & 99.20 & 100.00 & 100.00 & 10000 & 0.995 & 0.756 & 0.755 & 0.755 \\
Udiv & 7 & 1 & 126 & 1000 & 0.00 & 31.70 & 100.00 & 100.00 & 10000 & 1.000 & 0.033 & 0.006 & 0.006 \\
UdivExact & 5 & 2 & 95 & 1000 & 0.00 & 34.30 & N/A & \textbf{34.30} & 0 & - & - & - & - \\
Umax & 1 & 0 & 13 & 1000 & 0.00 & 100.00 & 100.00 & 100.00 & 10000 & 0.996 & 0.839 & 0.839 & 0.839 \\
Umin & 1 & 0 & 15 & 1000 & 0.00 & 100.00 & 100.00 & 100.00 & 10000 & 0.996 & 0.839 & 0.839 & 0.839 \\
UshlSat & 5 & 1 & 85 & 1000 & 0.00 & 100.00 & 100.00 & 100.00 & 10000 & 1.000 & 0.000 & 0.000 & 0.000 \\
UsubSat & 2 & 0 & 31 & 1000 & 0.00 & 100.00 & 100.00 & 100.00 & 10000 & 0.995 & 0.753 & 0.753 & 0.753 \\
Xor & 13 & 2 & 232 & 1000 & 50.90 & 56.10 & 66.70 & \textbf{68.20} & 10000 & 0.921 & 0.908 & 0.890 & \textbf{0.887} \\
\bottomrule
\end{tabular}
\end{table}

%% file: tables/compilation-time-table.tex

\begin{table}[t]
\centering
\caption{Comparison of compilation times and known bits precision
  between LLVM's own transfer functions and ours, for
  the nine SPEC CPU 2017 integer benchmarks in C/C++}
\label{tab:compile-time}
\footnotesize
\setlength{\tabcolsep}{2pt}
\renewcommand{\arraystretch}{1.05}

\begin{tabular}{r@{}r{r}*{30}{r{r}}@{}}
& & perlbench & gcc & mcf&omnetpp &xalancbmk &x264 &deepsjeng &leela &xz\\
\midrule[\lightrulewidth]
& KLOC & 362 & 1,304 & 3  & 134 & 520 & 96 & 10 & 21 & 33 \\
\midrule[\lightrulewidth]
& LLVM & 68.39 &
370.60 &
1.94 &
117.82& 
303.47 &
42.08 &
4.78 &
13.49 &
11.31 &
\\
\textbf{Compile time (s)} & Ours & 69.95 &
376.80 &
2.05 &
117.99& 
304.61 &
46.41 &
5.07 &
13.66 &
11.74 &
 \\
& Slowdown & 
2.27\,\% &
1.67\,\% &
5.49\,\% &
0.14\,\% &
0.38\,\% &
10.29\,\% &
5.89\,\% &
1.22\,\% &
3.80\,\% &
 \\
\midrule[\lightrulewidth]
& LLVM &
1,356,555&
4,272,154&
910&
62,251&
475,736&
247,344&
15,578&
25,207&
76,907&
\\
\textbf{Known bits} & Ours &
1,305,537&
4,195,918&
910&
62,102&
442,838&
218,171&
14,780&
19,353&
72,090&
\\
& Precision loss &
3.76\,\% &
1.78\,\% &
0.00\,\% &
0.24\,\% &
6.92\,\% &
11.79\,\% &
5.12\,\% &
23.22\,\% &
6.26\,\% &
\end{tabular}

\vspace{-1pt}
\end{table}

%% file: tables/intrinsics.tex
\begin{table}[tbhp]
\centering
\caption{Results for \textit{top}, \textit{composed}, and
  \textit{synth} on unary (6{,}561 test cases) and binary
  (43{,}046{,}721 cases) functions. Exact is reported in \%
  (higher is better), and precision is reported as a normalized count
  (lower is better).}
\label{tab:intrinsics}
\footnotesize
\setlength{\tabcolsep}{5pt}
\begin{tabular}{@{}l
                l
                rrr
                rrr@{}}
\toprule
\multirow{2}{*}{\textbf{Category}} & \multirow{2}{*}{\textbf{Concrete Op}} & \multicolumn{3}{c}{\textbf{Exact (\%)} $\uparrow$} & \multicolumn{3}{c}{\textbf{Precision (norm)} $\downarrow$} \\
\cmidrule(lr){3-5} \cmidrule(lr){6-8}
 & & $\top$ & composed & synth & $\top$ & composed & synth \\
\midrule 
 & Abs         & 1.95 & 3.92 & \textbf{100.00} & 3372 & 2552 & \textbf{0} \\
\emph{Unary Functions} & CountRZero  & 0.00 & 33.33 & \textbf{83.63}  & 5740 & 3553 & \textbf{193} \\
(6{,}561 test cases) & CountLZero  &0.00 & 0.00 & \textbf{83.63}  & 5740 & 5466 & \textbf{193} \\
 & PopCount    & 0.00 & 0.05 & \textbf{69.53}  & 4461 & 4456 & \textbf{369} \\
\midrule
 & Smax        & 4.46 & 6.33 & \textbf{56.86} & 19{,}797{,}600 & 18{,}179{,}000 & \textbf{5{,}471{,}520} \\
\emph{Binary Functions}  & Smin        & 4.46 & 6.04 & \textbf{70.39} & 19{,}797{,}600 & 18{,}384{,}100 & \textbf{4{,}117{,}500} \\
(43{,}046{,}721 test cases) & UaddSat     & 13.93 & 22.93 & \textbf{56.20} & 17{,}358{,}900 & 14{,}111{,}200 & \textbf{4{,}471{,}530} \\
 & UsubSat     & 13.93 & 19.46 & \textbf{49.11} & 17{,}358{,}900 & 15{,}377{,}400 & \textbf{5{,}369{,}170} \\
\bottomrule
\end{tabular}
\end{table}

%% file: tables/rp-table.tex
\begin{table}[t]
\centering
\caption{Results for reduced product between \knownbits and \constantrange for 8-bit and 64-bit integers. The columns and cells have the same meaning as in~\Cref{tab:kb-table}. Only operations for which the reduced product has an improvement over our synthesized \knownbits transformer are included.}
\label{tab:rp-table}
\footnotesize
\setlength{\tabcolsep}{5pt}
\begin{tabular}{@{}|l
                r
                rrr
                |r
                rrr@{}}
\toprule
\multirow{2}{*}{\textbf{Concrete Op}} & \multirow{2}{*}{\textbf{Tests}} & \multicolumn{3}{c}{\textbf{8-bit exact (\%)} $\uparrow$} & \multirow{2}{*}{\textbf{Tests}} & \multicolumn{3}{c}{\textbf{64-bit precision (norm)} $\downarrow$} \\
\cmidrule(lr){3-5} \cmidrule(lr){7-9}
 & & $\top$ & synth & reduced & & $\top$ & synth & reduced \\
\midrule
Abds         & 1000 & 7.64  & 18.06 & \textbf{28.98}  & 10000 &  0.1260  & 0.1119 & \textbf{0.1069} \\
Abdu         & 1000 & 7.11  & 20.76 & \textbf{71.06}  & 10000 &  0.1236  & 0.1085 & \textbf{0.0921} \\
AddNsw       & 1000 & 6.68  & 18.73 & \textbf{81.65}  & 10000 &  0.2820  & 0.1125 & \textbf{0.0869} \\
AddNswNuw    & 1000 & 0.22  & 44.89 & \textbf{88.20}  & 10000 &  0.5592  & 0.1216 & \textbf{0.0610} \\
AddNuw       & 1000 & 3.35  & 31.80 & \textbf{92.32}  & 10000 &  0.4920  & 0.0930 & \textbf{0.0557} \\
AvgCeilS     & 1000 & 9.83  & 18.01 & \textbf{29.16}  & 10000 &  0.1651  & 0.1573 & \textbf{0.1503} \\
AvgFloorS    & 1000 & 9.86  & 17.37 & \textbf{46.61}  & 10000 &  0.1669  & 0.1598 & \textbf{0.1412} \\
AvgFloorU    & 1000 & 9.90  & 19.00 & \textbf{50.37}  & 10000 &  0.1668  & 0.1481 & \textbf{0.1336} \\
Sdiv         & 1000 & 17.99 & 27.85 & \textbf{45.61}  & 10000 &  0.7262  & 0.2493 & \textbf{0.2229} \\
Smax         & 1000 & 0.44  & 59.89 & \textbf{83.19}  & 10000 &  0.4959  & 0.0926 & \textbf{0.0822} \\
Smin         & 1000 & 0.43  & 59.62 & \textbf{84.23}  & 10000 &  0.4954  & 0.0953 & \textbf{0.0843} \\
Srem         & 1000 & 13.14 & 22.41 & \textbf{26.92}  & 10000 &  0.1845  & 0.1701 & \textbf{0.1689} \\
SshlSat      & 1000 & 4.08  & 33.43 & \textbf{43.35}  & 10000 &  0.9542  & 0.3211 & \textbf{0.3148} \\
SubNswNuw    & 1000 & 0.33  & 39.01 & \textbf{77.20}  & 10000 &  0.5617  & 0.1299 & \textbf{0.0701} \\
SubNuw       & 1000 & 3.52  & 36.84 & \textbf{90.94}  & 10000 &  0.4822  & 0.0763 & \textbf{0.0466} \\
UaddSat      & 1000 & 4.11  & 61.03 & \textbf{83.09}  & 10000 &  0.4482  & 0.0874 & \textbf{0.0469} \\
Udiv         & 1000 & 0.00  & 68.66 & \textbf{75.28}  & 10000 &  0.9845  & 0.0134 & \textbf{0.0067} \\
UdivExact    & 1000 & 0.02  & 3.21  & \textbf{5.78}   & 10000 &  1.0000  & 0.0272 & \textbf{0.0195} \\
Umax         & 1000 & 0.54  & 95.28 & \textbf{99.74}  & 10000 &  0.4947  & 0.0016 & \textbf{0.0001} \\
Umin         & 1000 & 0.56  & 92.99 & \textbf{99.59}  & 10000 &  0.4964  & 0.0023 & \textbf{0.0003} \\
Urem         & 1000 & 2.12  & 61.45 & \textbf{66.53}  & 10000 &  0.2677  & 0.0393 & \textbf{0.0367} \\
UsubSat      & 1000 & 4.06  & 56.03 & \textbf{73.09}  & 10000 &  0.4508  & 0.1106 & \textbf{0.0700} \\
\bottomrule
\end{tabular}
\end{table}

%% file: 8related_work.tex
\section{Other Related Work}
\label{se:related-work}

Our approach draws on a rich line of work on synthesizing abstract transformers, verification infrastructure, and stochastic program synthesis. We build on the MLIR ecosystem~\cite{mlir-pldi25}, expressing synthesized abstract transformers in a first-class dialect that supports both efficient compilation via LLVM and formal reasoning via SMT encoding, using Z3~\cite{z3} for soundness verification.
\changed{We have already discussed at length how \name relates to its closest related work, \amurth~\cite{DBLP:journals/pacmpl/KalitaMDRR22,10.1007/978-3-031-74776-2_6}, in \Cref{se:evaluation:amurth}, and use the rest of the section to discuss other approaches.}

\subsubsection*{Precision-Oriented Synthesis and Domain-Specific Approaches}

Several efforts have addressed the problem of deriving best or approximate abstract transformers.
\changed{A line of work by Reps, Sagiv, Yorsh, and Thakur~\cite{VMCAI:RSY04,DBLP:conf/cav/ThakurR12,DBLP:conf/sas/ThakurER12,ThakurLLR15,VMCAI:Reps16} develops methods for automated symbolic abstraction—computing the best abstract transformer for a given input.
In these approaches, the transformer is not necessarily an explicit, executable program.
Other works have focused on finding executable representations of abstract transformers, but are typically tied to specific abstract domains or specific representations. 
For example, \citet{rr-asplos-2004} encode transformers using BDDs, which can sometimes be inefficient due to the limited expressivity of BDDs.
\citet{DBLP:journals/toplas/ElderLSAR14} target conjunctions of bit-vector equalities, 
\citet{llh-oopsla-2025} target numerical abstract domains such as linear convex polytope,
and \citet{harishankar2025ebpf} target interval analysis in the eBPF verifier.
In contrast, \name\ supports a wide range of instructions for composing abstract transformers and remains agnostic to specific domains.}
%
%
More recent work proposes sketching-based algorithms for learning disjunctive and conjunctive specifications over program behaviors~\cite{DBLP:journals/pacmpl/ParkDR23,loudfull}. While our work shares the idea of synthesizing multiple components and combining them (via meet), we instantiate it in the domain of abstract interpretation with formal guarantees and no sketching.

\subsubsection*{Stochastic and MCMC-Based Synthesis}

Our synthesis algorithm is inspired by stochastic search techniques, in particular the Markov Chain Monte Carlo (MCMC) superoptimization strategy introduced by Stoke~\cite{Alex2013Stoke}. Like Stoke, our framework searches the space of candidate programs guided by a cost function, using probabilistic rewrites. Unlike Stoke, which targets concrete program optimization, we use MCMC to synthesize abstract transformers, and our cost function encodes the precision improvement from existing transformers. Moreover, our work introduces a novel abductive refinement strategy that iteratively improves precision by synthesizing and composing multiple sound transformers. Instantiating this algorithm over the full LLVM instruction set via MLIR requires significant engineering and forms a core contribution of our work.

%% file: 10appendix-amurth.tex
\section{Detailed Comparison to \amurth}
\label{app:amurth}

\changed{\citeauthor{DBLP:journals/pacmpl/KalitaMDRR22} used \amurth to synthesize transformers for both string domains and fixed-bitwidth integer domains.
    We focus our comparison on integer domains, as \name currently supports only integer domains.
    \amurth synthesized transformers for 9 concrete operators for unsigned and signed interval domains: \texttt{add}, \texttt{sub}, \texttt{mul}, \texttt{and}, \texttt{or}, \texttt{xor}, \texttt{shl}, \texttt{ashr}, and \texttt{lshr}~\cite[Table 4]{DBLP:journals/pacmpl/KalitaMDRR22}.
    These operators are also featured in our benchmark set.
    \amurth successfully synthesized the best transformers for each of these operators within 30 minutes, whereas \name synthesized the best transformers only for \texttt{add}, \texttt{sub}, and \texttt{ashr}.}

\changed{However, the reason \amurth performs well is that it requires users to supply program templates and auxiliary functions, which serve as powerful hints to guide synthesis.
    Moreover, these hints differ across synthesis tasks.
    When \amurth is restricted to the same setting as \name—i.e., using only the base DSL operators and no templates—it times out on all benchmarks for both unsigned and signed interval domains.
    In the following case studies, we examine several examples of such hints from the \amurth benchmark suite and illustrate why \amurth  is not suited to automate the synthesis of many transformers in a compiler.}

\paragraph{\amurth Requires Sketches or Templates}
\label{s:evaluation:amurth:template}
\changed{
    Users of \amurth typically need to provide a program sketch or template of the desired solution, even though the theoretical framework supporting \amurth is parametric in the choice of DSL. The sketches are often hand-derived from existing manual implementations~(noted in~\citet[Table 2 and Table 3]{DBLP:journals/pacmpl/KalitaMDRR22}).
}

\changed{
    For instance, the template used to synthesize the transformer for \texttt{xor} \citet[Figure 16]{DBLP:journals/pacmpl/KalitaMDRR22} in the signed domain begins with a helper function \texttt{splitAtZero}, which divides each input interval at 0 if it crosses 0.
    The core synthesis task then fills in the logic that transforms the resulting (up to 2*2=4) pairs of same-signed intervals, before the template finally joins those outputs into the final interval.
    While the template is effective, both the \texttt{splitAtZero} helper and the looping pattern over split-interval pairs are non-trivial to synthesize from scratch.
}

\changed{
    In contrast, our candidate program fixes only a minimal structure: it deconstructs each input interval into two integers at the beginning and reconstructs the output interval from two integers at the end.
    The middle portion is a free-form SSA program, unconstrained in its instruction dependencies.
}

\changed{
    Additionally, \amurth does not synthesize boundary conditions to handle overflow~\cite[Section 6.2.3]{DBLP:journals/pacmpl/KalitaMDRR22}.
    It synthesizes the non-overflow case, and such an assumption can be seen as an implicit template.
    For example, consider the best transformer for addition in the unsigned interval domain, shown in \Cref{fig:ucr-add}.
}
\changed{
    Our approach synthesized an equivalent transformer. In contrast, \amurth  synthesized only the else branch $[a.l+b.l, \; a.r+b.r]$, as it assumes its output will be plugged into a sketch that checks for overflow. Such an assumption is equivalent to $(a.l+b.l > \texttt{MAX\_INT})\; \textbf{or}\; (a.r+b.r > \texttt{MAX\_INT})$. This condition is weaker than the one in our transformer because it means an overflow for additions of either the left \textbf{or} right endpoints (which should be an \textbf{xor}).
    As a result, their no-overflow assumption not only reduces the difficulty of the synthesis, but also loses precision on the overflow cases.
}
\input{codes/cr_add.tex}

\paragraph{\amurth relies on auxiliary functions}
\changed{
    This limitation is reflected in~\cite[Section 6.2.4]{DBLP:journals/pacmpl/KalitaMDRR22} and we confirmed it by examining the Amurth codebase.
    For example, to synthesize transformers for bitwise operators (\texttt{and}, \texttt{or}, \texttt{xor}), auxiliary functions such as \texttt{minOr}, \texttt{maxOr}, \texttt{minAnd}, \texttt{maxAnd} (which compute the lower/upper bound of the results of bitwise-or/and over 2 intervals) are provided.
}

\changed{
    These auxiliary functions are non-trivial, consisting of about 25 operators and involving branching and loops, and make the synthesis tasks much easier. Even with the hints above, \amurth still needs advanced program sketches for several harder benchmarks.
}

\paragraph{Hints vary across synthesis tasks}
\changed{
    When provided with enough structure and templates, \amurth can directly synthesize optimal transformers.
    However, the supporting DSL is typically crafted and modified individually for synthesis tasks of each concrete operator. For example, even within the same abstract domain, bitwise operators, arithmetic operators, and shifting operators \emph{each rely on a distinct DSL}. These DSLs consist of nearly disjoint sets of operations.
}

\changed{
    In contrast, our approach employs a single, unified DSL shared across all concrete operators.
    It consists of 29 basic numeric operations that can express a wide variety of transformers.
    Details of the DSL’s design and its operations are discussed in~\Cref{se:dsl}.
}

\paragraph{Summary.}
\changed{
    To summarize, one cannot use \amurth to solve the problem solved by \name, i.e., automatically synthesizing transformers for many operators at once without manually tuning the underlying DSL or providing strong hints in the form of sketches.
    While \amurth does appear to be better than \name for the task of finding tricky transformers for individual operators, it needs guidance from human experts.
    In practice, for an abstract domain, there could be hundreds of operators that need an abstract transformer, and providing tailored hints for each of them is infeasible.
    Hence, we believe that \amurth cannot be used to automate large-scale synthesis of transformers in a compiler.
}

%% file: codes/cr_add.tex
\begin{figure}[h]
\vspace{-10pt}
\begin{subfigure}[b]{0.57\columnwidth}
    \begin{equation*}
        \begin{array}{l}
            \texttt{add}^\sharp(a,b) :=                                       \\
            \quad \textbf{if } (a.l+b.l > \texttt{MAX\_INT})\; \textbf{xor}\;  (a.r+b.r > \texttt{MAX\_INT})                              \\
            \quad \textbf{then } [0, \texttt{MAX\_INT}]                       \\
            \quad \textbf{else } [a.l+b.l, \; a.r+b.r]
        \end{array}
    \end{equation*}
    \captionsetup{justification=raggedright,singlelinecheck=false}
    \caption{The best transformer for the \texttt{add} operator\\ in $\ucr$ domain}
    \vspace{-10pt}
    \label{fig:ucr-add}
\end{subfigure}%
\begin{subfigure}[b]{0.425\columnwidth}
\begin{lstlisting}[language=llvm, numbers=none]
KnownBits add#(KnownBits L, R) { 
  APInt l0 = L.Zero, l1 = L.One;
  APInt r0 = r.Zero, r1 = R.One;
  APInt E = (l0 | l1) & (r0 | r1) &
    (~((~l0 + ~r0) ^ l0 ^ r0) | 
    ((l1 + r1) ^ l1 ^ r1));
  APInt known0 = ~(~l0 + ~r0) & E;
  APInt known1 = (l1 + r1) & E;
  return {known0, known1};
}
\end{lstlisting}
    \vspace{-10pt}
\caption{The best transformer for the \texttt{add} operator in the \knownbits domain}
    \vspace{-10pt}
\label{fig:best-add}
\end{subfigure}
\caption{Optimal transformers}
\end{figure}

%% file: 11appendix-ablation.tex
\section{Ablation Study}
\label{app:ablation}

\subsection{Impact of DSL Choice}
\label{app:ablation:dsl}
\input{codes/best-add}
\changed{
We evaluate how the choice of operations in DSL affects the performance of \name.
As mentioned in \Cref{se:dsl}, we run synthesis over the \textbf{Full} language and 2 subsets (\basicops and \bitops).
Detailed results are summarized in \Cref{tab:dsl-table}.
}

\changed{\basicops is the DSL with the most limited expressivity and causes substantial precision loss in many benchmarks, though it produces optimal transformers for \texttt{add} and \texttt{sub}.
\Cref{fig:best-add} shows a reference implementation of the \texttt{add} transformer in \knownbits, which uses only operations from \basicops but remains fairly complex.
Nevertheless, \name successfully synthesizes a transformer equivalent to this implementation when the DSL is limited to \basicops.
}

\changed{\bitops is the DSL variant that excludes multiplicative operations.
    Since the transformers for shifting operators (e.g.,~\texttt{shl}) rarely rely on multiplication and division,
    synthesis with \bitops achieves slightly higher precision for these operators compared to using the full DSL.}

\changed{In conclusion, if certain operations are known to be irrelevant for a specific transformer, removing them from the DSL can improve the precision.}
\input{tables/dsl-table}

\subsection{Impact of Condition Abduction}
\label{app:ablation:abduction}

\changed{In this section, we evaluate the impact of condition abduction (\Cref{sec:abduction}).}
\changed{We run synthesis with and without condition abduction on both the \knownbits and \constantrange domains.
    For \knownbits, condition abduction improves precision by $6.44\%$ on average (geometric mean) across all benchmarks, with 18/39 benchmarks showing gains.
    For \constantrange, condition abduction improves precision by $2.3\%$ on average (geometric mean) across all benchmarks, with 16/39 benchmarks showing gains.
    Results are shown in Tables~\ref{tab:kb-abduction} and~\ref{tab:cr-abduction}.
}

\changed{A representative case highlighting the necessity of abduction is the \texttt{add} benchmark in the $\ucr$ domain.
    Its best transformers, shown in \Cref{fig:ucr-add}, produce intervals better than $\top$ only when whether overflow happens is known.
    When inspecting the synthesis process, we observe that \name synthesized the best transformer in 2 rounds:
    in the first round, it discovered a transformer that only works for non-overflow cases and stored it as one of the unsound but highly precise candidates;
    in the second round, it successfully identified the overflow condition the complete the full transformers.
    However, when abduction is disabled, \name failed to synthesize the best transformer within 5 rounds.}

\changed{To summarize, condition abduction generally improves precision, but may sometimes reduce it since condition abduction reuses part of the parallel search budget.
    However, one can also run \name with and without abduction and selecting the better result.}
\input{tables/kb-abduction-table}
\input{tables/cr-abduction-table}

%% file: codes/best-add.tex

%% file: tables/dsl-table.tex
\newcolumntype{C}[1]{>{\centering\arraybackslash}p{#1}}
\newcommand{\cw}{6.4mm}    
\newcommand{\fw}{11mm}   
\newcommand{\bench}[1]{\rotatebox{55}{\scriptsize\texttt{#1}}}

\begin{table}[t]
\centering
\caption{
\changed{\knownbits results under DSLs with different operation subsets.
Each value denotes the percentage of tests where the synthesized transformer matches the best transformer $\besttf$ at 8-bit precision.
\textbf{Bolded} numbers denote the configuration that yields the most precise result}
}
\label{tab:dsl-table}
\footnotesize
\setlength{\tabcolsep}{0pt}
\renewcommand{\arraystretch}{1.05}

\begin{tabular}{@{}C{\fw}*{20}{C{\cw}}@{}}
\toprule
 &
\bench{Abds} & \bench{Abdu} & \bench{Add} & \bench{AddNsw} & \bench{AddNswNuw} & \bench{AddNuw} & \bench{And} & \bench{Ashr} & \bench{AshrExact} & \bench{AvgCeilS} & \bench{AvgCeilU} & \bench{AvgFloorS} & \bench{AvgFloorU} & \bench{Lshr} & \bench{LshrExact} & \bench{Mods} & \bench{Modu} & \bench{Mul} & \bench{Or} & \bench{Sdiv} \\
\midrule
\textbf{Full} & 60.1 & \textbf{59.4} & 58.7 & 42.0 & 45.5 & \textbf{53.9} & \textbf{100.0} & \textbf{65.5} & \textbf{40.1} & 38.8 & 38.6 & \textbf{39.3} & 37.7 & \textbf{59.3} & 31.4 & \textbf{64.7} & \textbf{59.0} & 60.6 & \textbf{100.0} & \textbf{64.1} \\
\basicops & 54.7 & 53.3 & \textbf{100.0} & \textbf{78.0} & 12.9 & 37.9 & 100.0 & 31.2 & 30.0 & 31.8 & 31.6 & 32.4 & 32.2 & 12.6 & 28.1 & 48.6 & 39.3 & \textbf{74.4} & 100.0 & 64.1 \\
\bitops & \textbf{60.3} & 59.3 & 51.8 & 51.3 & \textbf{49.0} & 49.5 & 100.0 & 37.8 & 36.6 & \textbf{39.4} & \textbf{40.9} & 39.0 & \textbf{40.5} & 31.0 & \textbf{31.6} & 60.4 & 58.4 & 62.1 & 100.0 & 64.1 \\
\bottomrule
\end{tabular}

\vspace{-1pt}

\begin{tabular}{@{}C{\fw}*{20}{C{\cw}}@{}}
\toprule
 &
\bench{SdivExact} & \bench{Shl} & \bench{ShlNsw} & \bench{ShlNswNuw} & \bench{ShlNuw} & \bench{Smax} & \bench{Smin} & \bench{SshlSat} & \bench{Sub} & \bench{SubNsw} & \bench{SubNswNuw} & \bench{SubNuw} & \bench{UaddSat} & \bench{Udiv} & \bench{UdivExact} & \bench{Umax} & \bench{Umin} & \bench{UshlSat} & \bench{UsubSat} & \bench{Xor} \\
\midrule
\textbf{Full} & 19.3 & \textbf{56.9} & 26.2 & 9.8 & 40.5 & 63.8 & \textbf{72.8} & \textbf{72.4} & 60.6 & 47.8 & \textbf{31.8} & \textbf{47.1} & 61.8 & \textbf{80.8} & 15.3 & 90.6 & \textbf{92.9} & \textbf{96.6} & 52.1 & \textbf{100.0} \\
\basicops & 19.3 & 13.2 & 6.9 & 5.6 & 10.7 & 57.8 & 51.9 & 37.6 & \textbf{100.0} & \textbf{72.3} & 12.0 & 42.4 & 42.5 & 24.1 & 2.9 & 71.4 & 72.2 & 3.7 & 48.2 & 100.0 \\
\bitops & 19.3 & 56.6 & \textbf{27.9} & \textbf{34.9} & \textbf{46.9} & \textbf{71.0} & 56.1 & 45.0 & 61.5 & 51.4 & 28.5 & 45.0 & \textbf{67.4} & 39.9 & \textbf{15.4} & \textbf{93.3} & 88.3 & 96.6 & \textbf{56.9} & 100.0 \\
\bottomrule
\end{tabular}
\end{table}

%% file: tables/kb-abduction-table.tex

\begin{table}[h]
\centering
\caption{
\changed{\knownbits results with and without abduction.
Each value denotes the percentage of tests where the synthesized transformer matches the best transformer $\besttf$ at 8-bit precision.}
}
\label{tab:kb-abduction}
\footnotesize
\setlength{\tabcolsep}{0pt}
\renewcommand{\arraystretch}{1.05}

\begin{tabular}{@{}C{\fw}*{20}{C{\cw}}@{}}
\toprule
 &
\bench{Abds} & \bench{Abdu} & \bench{Add} & \bench{AddNsw} & \bench{AddNswNuw} & \bench{AddNuw} & \bench{And} & \bench{Ashr} & \bench{AshrExact} & \bench{AvgCeilS} & \bench{AvgCeilU} & \bench{AvgFloorS} & \bench{AvgFloorU} & \bench{Lshr} & \bench{LshrExact} & \bench{Mods} & \bench{Modu} & \bench{Mul} & \bench{Or} & \bench{Sdiv} \\
\midrule
Abd & 60.1 & 59.4 & 58.7 & 42.0 & 45.5 & \textbf{53.9} & \textbf{100.0} & \textbf{65.5} & \textbf{40.1} & \textbf{38.8} & \textbf{38.6} & 39.3 & 37.7 & \textbf{59.3} & \textbf{31.4} & \textbf{64.7} & \textbf{59.0} & 60.6 & \textbf{100.0} & 64.1 \\
No Abd & \textbf{60.7} & \textbf{61.7} & \textbf{63.4} & \textbf{48.8} & \textbf{50.5} & 40.3 & 100.0 & 49.3 & 25.0 & 38.8 & 38.4 & \textbf{40.3} & \textbf{40.2} & 59.3 & 24.4 & 64.7 & 44.6 & \textbf{64.3} & 100.0 & \textbf{67.8} \\
\bottomrule
\end{tabular}

\vspace{-1pt}

\begin{tabular}{@{}C{\fw}*{20}{C{\cw}}@{}}
\toprule
 &
\bench{SdivExact} & \bench{Shl} & \bench{ShlNsw} & \bench{ShlNswNuw} & \bench{ShlNuw} & \bench{Smax} & \bench{Smin} & \bench{SshlSat} & \bench{Sub} & \bench{SubNsw} & \bench{SubNswNuw} & \bench{SubNuw} & \bench{UaddSat} & \bench{Udiv} & \bench{UdivExact} & \bench{Umax} & \bench{Umin} & \bench{UshlSat} & \bench{UsubSat} & \bench{Xor} \\
\midrule
Abd & \textbf{19.3} & \textbf{56.9} & 26.2 & 9.8 & \textbf{40.5} & 63.8 & \textbf{72.8} & \textbf{72.4} & \textbf{60.6} & 47.8 & \textbf{31.8} & 47.1 & 61.8 & \textbf{80.8} & 15.3 & \textbf{90.6} & \textbf{92.9} & \textbf{96.6} & \textbf{52.1} & \textbf{100.0} \\
No Abd & 19.3 & 55.5 & \textbf{27.6} & \textbf{41.8} & 25.4 & \textbf{67.0} & 54.3 & 65.7 & 53.5 & \textbf{57.8} & 14.4 & \textbf{50.0} & \textbf{63.4} & 79.8 & \textbf{17.7} & 73.2 & 62.1 & 92.9 & 19.1 & 100.0 \\
\bottomrule
\end{tabular}
\end{table}

%% file: tables/cr-abduction-table.tex

\begin{table}[h]
\centering
\caption{\constantrange results with and without abduction.
The cells have the same meaning as in~\Cref{tab:kb-abduction}.}
\label{tab:cr-abduction}
\footnotesize
\setlength{\tabcolsep}{0pt}
\renewcommand{\arraystretch}{1.05}

\begin{tabular}{@{}C{\fw}*{19}{C{\cw}}@{}}
\toprule
 &
\bench{Abds} & \bench{Abdu} & \bench{Add} & \bench{AddNsw} & \bench{AddNswNuw} & \bench{AddNuw} & \bench{And} & \bench{Ashr} & \bench{AshrExact} & \bench{AvgCeilS} & \bench{AvgCeilU} & \bench{AvgFloorS} & \bench{AvgFloorU} & \bench{Lshr} & \bench{LshrExact} & \bench{Mods} & \bench{Modu} & \bench{Mul} & \bench{Or} \\
\midrule
Abd & \textbf{59.8} & \textbf{75.0} & \textbf{100.0} & \textbf{100.0} & \textbf{70.7} & \textbf{94.0} & \textbf{82.3} & \textbf{98.0} & \textbf{81.7} & 0.3 & 1.8 & 1.8 & 30.0 & \textbf{100.0} & \textbf{85.2} & 43.5 & \textbf{89.0} & \textbf{90.8} & 82.4 \\
No Abd & 59.8 & 67.2 & 36.5 & 68.6 & 61.5 & 81.7 & 76.3 & 95.3 & 81.0 & \textbf{0.7} & \textbf{2.9} & \textbf{2.0} & \textbf{67.8} & 100.0 & 72.5 & \textbf{43.9} & 89.0 & 90.8 & \textbf{82.5} \\
\bottomrule
\end{tabular}

\vspace{-1pt}

\begin{tabular}{@{}C{\fw}*{19}{C{\cw}}@{}}
\toprule
 &
\bench{Sdiv} & \bench{Shl} & \bench{ShlNswNuw} & \bench{ShlNuw} & \bench{Smax} & \bench{Smin} & \bench{SshlSat} & \bench{Sub} & \bench{SubNsw} & \bench{SubNswNuw} & \bench{SubNuw} & \bench{UaddSat} & \bench{Udiv} & \bench{Umax} & \bench{Umin} & \bench{UshlSat} & \bench{UsubSat} & \bench{Xor} & {} \\
\midrule
Abd & \textbf{12.3} & \textbf{83.0} & \textbf{88.3} & 76.4 & \textbf{100.0} & \textbf{100.0} & \textbf{75.5} & \textbf{100.0} & \textbf{100.0} & \textbf{74.5} & 94.5 & \textbf{99.2} & 31.7 & \textbf{100.0} & \textbf{100.0} & \textbf{100.0} & \textbf{100.0} & 56.1 & {} \\
No Abd & 3.7 & 83.0 & 88.3 & \textbf{87.7} & 100.0 & 100.0 & 75.2 & 35.4 & 70.1 & 59.7 & \textbf{100.0} & 85.6 & \textbf{40.5} & 100.0 & 100.0 & 98.6 & 98.8 & \textbf{60.6} & {} \\
\bottomrule
\end{tabular}

\end{table}